\documentclass[conference]{IEEEtran}

\makeatletter
\newcommand\semihuge{\@setfontsize\semihuge{22.3}{22}}
\makeatother

\usepackage{algpseudocode}
\usepackage{algorithm}
\usepackage{algorithmicx}
\usepackage{latexsym}
\usepackage{lipsum} 
\usepackage[dvips]{color}
\usepackage{comment}
\usepackage{todonotes}
\usepackage{epsf}
\usepackage{times}
\usepackage{epsfig}
\usepackage{graphicx}
\usepackage{bbold}
\usepackage{mathtools}
\usepackage{mathrsfs}
\usepackage{amssymb}
\usepackage{pdfpages}
\usepackage{epstopdf}
\usepackage{float}
\usepackage{amsfonts}
\usepackage{here}
\usepackage{rawfonts}
\usepackage{dsfont}
\usepackage{lettrine} 
\usepackage{url}
\usepackage{amsthm}
\usepackage{cite}
\allowdisplaybreaks
\usepackage{csquotes}
\usepackage{verbatim}
\usepackage[english]{babel}
\usepackage{xfrac}

\usepackage{amsmath}

\newfloat{algorithm}{t}{lop}

\ifCLASSOPTIONcompsoc
    \usepackage[caption=false, font=normalsize, labelfont=sf, textfont=sf]{subfig}
\else
\usepackage[caption=false, font=footnotesize]{subfig}
\fi

\newtheorem{theorem}{\bf Theorem}

\newtheorem{proposition}{\bf Proposition}

\newcommand\blfootnote[1]{%
  \begingroup
  \renewcommand\thefootnote{}\footnote{#1}%
  \addtocounter{footnote}{-1}%
  \endgroup
}
	
\DeclareMathOperator*{\argmax}{\arg\!\max}

\begin{document}\bstctlcite{IEEEexample:BSTcontrol}
\title{\LARGE Centralized and Distributed Age of Information Minimization with non-linear Aging Functions in the Internet of Things\vspace{-0.3cm}}    

\author{\IEEEauthorblockN{Taehyeun Park$^1$, Walid Saad$^1$, and Bo Zhou$^1$}\vspace{-0.05cm}\\
	\IEEEauthorblockA{\small $^1$Wireless@VT, Bradley Department of Electrical and Computer Engineering, Virginia Tech, Blacksburg, VA, USA,\\ Emails:\{taehyeun, walids, ecebo\}@vt.edu}
	\vspace{-0.92cm}}
\maketitle\vspace{-0.8cm}
\thispagestyle{plain}
\pagestyle{plain}
\vspace{-0.1cm}
\begin{abstract}
Resource management in Internet of Things (IoT) systems is a major challenge due to the massive scale and heterogeneity of the IoT system. For instance, most IoT applications require timely delivery of collected information, which is a key challenge for the IoT. In this paper, novel centralized and distributed resource allocation schemes are proposed to enable IoT devices to share limited communication resources and to transmit IoT messages in a timely manner. In the considered system, the timeliness of information is captured using non-linear age of information (AoI) metrics that can naturally quantify the freshness of information. To model the inherent heterogeneity of the IoT system, the non-linear aging functions are defined in terms of IoT device types and message content. To minimize AoI, the proposed resource management schemes allocate the limited communication resources considering AoI. In particular, the proposed centralized scheme enables the base station to learn the device types and to determine aging functions. Moreover, the proposed distributed scheme enables the devices to share the limited communication resources based on available information on other devices and their AoI. The convergence of the proposed distributed scheme is proved, and the effectiveness in reducing the AoI with partial information is analyzed. Furthermore, the proposed resource management schemes with different number of devices, activation probabilities, and outage probabilities are analyzed in terms of the average instantaneous AoI. Simulation results show that the proposed centralized scheme achieves significantly lower average instantaneous AoI when compared to simple centralized allocation without learning, while the proposed distributed scheme achieves significantly lower average instantaneous AoI when compared to random allocation. The results also show that the proposed centralized scheme outperforms the proposed distributed scheme in almost all cases, but the distributed approach is more viable for a massive IoT.
\end{abstract}


\section{Introduction} \label{sec:intro}
\blfootnote{\noindent This research was supported by the Office of Naval Research (ONR) under MURI Grant N00014-19-1-2621.} The Internet of Things (IoT) is arguably the most important technology of the coming decade \cite{saad2019vision}. However, the effective operation of several IoT services, such as industrial monitoring \cite{indus}, health monitoring \cite{health}, drones \cite{mohadrone}, virtual reality \cite{chenvr}, and vehicular network \cite{vehicular}, requires timely and frequent communications. To maintain the proper performance of such diverse IoT applications, the base station (BS) must maintain the most relevant information gathered from the IoT devices at any given time.\\
\indent In addition to timely transmissions from the devices to the BS, another key challenge is to account for the distinctive characteristics of an IoT and its devices. One prominent property of an IoT is its massive scale as the number of devices greatly outnumbers the available communication resources \cite{saad}. Therefore, an appropriate allocation of the limited communication resources among numerous IoT devices is necessary for the deployment of an IoT and its services \cite{iotalloc}. Furthermore, the IoT exhibits a high heterogeneity in terms of device types, functions, messages, transmission requirements, and resource constraints \cite{minehetero}. The aforementioned IoT properties pose challenges for timely uplink transmission in an IoT. To ensure the performance of time-sensitive IoT applications despite the aforementioned challenges, a new information timeliness performance metric is needed as an alternative to conventional delay, reliability, and data rate.\\
\indent To evaluate the communication between the BS and the devices, the \textit{age of information} (AoI), which is a metric that can quantify the relevance and the freshness of the information, is used \cite{aoi1, bo2}. However, the AoI has different characteristics compared to delay \cite{aoi1}, because it explicitly considers packet generation time. The problem of AoI minimization in an IoT has unique challenges due to the characteristics of an IoT, including massive scale, limited communication resources, and IoT device heterogeneity. Largely, AoI minimization can be done in a centralized way or in a distributed way. However, a centralized AoI minimization approach is not always viable for an IoT, because the energy constrained IoT devices may not be able to communicate frequently with BS. On the other hand, a distributed AoI minimization approach may require extensive device-to-device communication and could perform worse than a centralized solution for some IoT scenarios. Therefore, both centralized and distributed AoI minimization must be investigated to compare their applicability and performances in an IoT.\vspace{-1mm}

\subsection{Existing Works}
A number of recent works studied the problem of AoI minimization in wireless networks \cite{bo2, aoi1, aoischedule1, aoischedule3, multihop, adhoc, aoimultiaccess, aoischedule2, aoicsi, bo3, aoischedule4, multiinfo, multisource, aoibackoff, aoicsma, aoisleep, qm3, mm1, mg1, mg11, aoiurllc, qm2, aoinoma, nonlin1, gc2020}. These studies use various approaches to minimize the AoI under different constraints and conditions. For instance, the works in \cite{bo2, aoi1, aoischedule1, aoischedule3, multihop, adhoc, aoimultiaccess, aoischedule2, aoicsi, bo3, aoischedule4, multiinfo, multisource, aoibackoff, aoicsma, aoisleep,qm3, mm1, mg1, mg11, aoiurllc, qm2, aoinoma, nonlin1} study a variety of scheduling policies for AoI minimization in different networks, including single hop broadcast network \cite{aoischedule1}, single-hop uplink communication \cite{aoischedule3}, multi-hop uplink communication \cite{multihop}, ad-hoc networks \cite{adhoc}, and ALOHA-like random access \cite{aoimultiaccess}. The authors in \cite{aoischedule2} and \cite{aoicsi} propose and analyze scheduling policies for the wireless networks with known and unknown channel state information. The works in \cite{aoi1}, \cite{aoischedule3}, \cite{adhoc}, and \cite{bo3} introduce effective scheduling policies to minimize the average AoI with network constraints, such as throughput requirement, physical constraint for sensing, spectrum sharing, and varying packet sizes. In \cite{adhoc}, \cite{aoischedule4}, and \cite{multiinfo}, the authors use online techniques, such as reinforcement learning to perform AoI-minimal scheduling. Moreover, the authors in \cite{multiinfo} and \cite{multisource} analyze the performance of user scheduling for minimizing the average AoI in presence of multiple sources of information and propose a hybrid queueing system. The authors in \cite{aoibackoff} analyze the coexistence of DSRC and WiFi networks as a game, in which the DSRC network minimizes the AoI and the WiFi network maximizes the throughput. For CSMA networks, the work in \cite{aoicsma} optimizes the backoff time of each communication link to minimize the total average AoI, and the authors in \cite{aoisleep} propose a sleep-wake scheduling to optimize the tradeoff between AoI minimization and energy consumption.\\
\indent The works in \cite{qm3, mm1, mg1, mg11} address the problem of AoI minimization using queueing-theoretic approaches. In \cite{qm3}, the authors analyze the peak AoI in a multi-class queueing system with packets having heterogeneous service times and requirements. The authors in \cite{mm1, mg1, mg11} derive closed-form solutions for the average AoI and the peak AoI for different queueing models, including M/M/1, M/G/1, and M/G/1/1. In \cite{aoiurllc}, the authors consider a vehicular network with ultra-reliable low-latency communication and minimize the tail of the AoI distribution. The peak AoI considering the packet delivery failure is analyzed in \cite{qm2}. Moreover, the non-orthogonal multiple access is compared against the conventional orthogonal multiple access in terms of AoI minimization in \cite{aoinoma}. The authors study the sampling policies to minimize the average AoI with the joint status sampling in IoT \cite{bo2} or with the non-linear aging functions \cite{nonlin1}.\\
\indent Despite being interesting, the existing solutions in \cite{bo2, aoi1, aoischedule1, aoischedule3, multihop, adhoc, aoimultiaccess, aoischedule2, aoicsi, bo3, aoischedule4, multiinfo, multisource, aoibackoff, aoicsma, aoisleep, qm3, mm1, mg1, mg11, aoiurllc, qm2, aoinoma, nonlin1, gc2020} do not consider some of the unique properties of an IoT, such as limited communication resources, massive scale, and high device heterogeneity. One of the key challenges in IoT is the massive scale of IoT coupled with highly limited available communication resources. However, the works in \cite{aoi1, aoischedule1, aoischedule3, multihop, adhoc, aoimultiaccess, aoischedule4}, and \cite{qm2, aoinoma, nonlin1} do not investigate the realistic IoT scenario in which the number of devices greatly outnumbers the communication resources. Furthermore, the inherent heterogeneity among IoT devices and the presence of non-linear aging functions are not considered in \cite{bo2, aoischedule1, aoischedule3, multihop, adhoc, aoimultiaccess, aoischedule2, aoicsi, bo3, aoischedule4}, and \cite{mm1, mg1, mg11, aoiurllc, qm2, aoinoma}. Moreover, most of the prior works for AoI minimization in \cite{aoi1, aoischedule1, aoischedule3, multihop, adhoc, aoimultiaccess, aoischedule2, aoicsi, multiinfo, multisource,aoibackoff, aoicsma, aoisleep,qm3, mm1, mg1, mg11}, and \cite{nonlin1} only considers a centralized approach. Moreover, in \cite{gc2020}, we studied centralized AoI minimization with non-linear aging functions and proposed a centralized resource allocation scheme to enable the BS to consider different aging functions. However, the work in \cite{gc2020} only investigates a centralized approach for AoI minimization and does not introduce a distributed resource allocation framework for an IoT with non-linear aging functions. A centralized approach may not always be suitable for an IoT, because the frequent communication with BS is not viable for the energy constrained IoT devices. These important challenges for enhancing the AoI in an IoT have been largely overlooked in prior works \cite{bo2, aoi1, aoischedule1, aoischedule3, multihop, adhoc, aoimultiaccess, aoischedule2, aoicsi, bo3, aoischedule4, multiinfo, multisource, aoibackoff, aoicsma, aoisleep, qm3, mm1, mg1, mg11, aoiurllc, qm2, aoinoma, nonlin1, gc2020}.
\subsection{Contributions}
\indent The main contributions of this paper are novel centralized and distributed resource allocation frameworks that can be used to minimize the average instantaneous AoI for a massive IoT with heterogeneous devices and non-linear aging functions. In particular, we capture the heterogeneity among IoT devices using \emph{non-linear aging functions}. Typically, the AoI is defined only in terms of time, and it is assumed to increase linearly with a slope of $1$ \cite{aoi1}. However, the definition of the AoI can be broader such that the AoI can be a function of completeness, validity, accuracy, currency, and utility \cite{nonlin2, nonlin3}. Under such a broader definition of the AoI, the aging function can be defined as an age penalty function or an age utility function \cite{nonlin1}, which can be an exponential, linear, or step function \cite{nonlin2}. As such, we propose to capture the heterogeneity among IoT devices and messages by assigning different aging functions based on the devices types, the IoT application, the message content, and the transmission requirement.\\
\indent For centralized AoI minimization, we propose a new priority scheduling scheme with a learning perspective such that the device types and the aging functions can be determined. For non-linear aging functions, we show that using the future AoI for priority scheduling achieves a lower average instantaneous AoI than using the current AoI. Simulation results for the centralized approach show that the proposed priority scheduling scheme achieves $26.7\%$ lower average instantaneous AoI with high activation probability and $31.7\%$ lower average instantaneous AoI with high outage probability than a simple priority scheduling. In particular, our approach outperforms a simple priority scheduling and performs similar to a priority scheduling with complete information on device types and aging functions.\\
\indent For the distributed AoI minimization, we formulate a minority game \cite{kolkata}, such that massive number of IoT devices can share the limited available communication resources autonomously. Furthermore, a payoff function is designed to allow the messages with the highest AoI to transmit first in a self-organizing manner. We then show the conditions that a resource allocation among IoT devices must satisfy to achieve a Nash equilibrium (NE). We propose a stochastic crowd avoidance algorithm for the resource allocation game and prove that the resource allocation using our proposed algorithm converges to an NE with sufficient information and under certain network parameters. Simulation results for the distributed case show that the proposed algorithm is effective in minimizing the AoI even if the devices only have the partial information about other devices. The results show that our game-based approach achieves $63.6\%$ lower average instantaneous AoI with limited information and $45.8\%$ lower average instantaneous AoI with high outage probability than a random resource allocation. Moreover, after convergence, our game-based approach performs similar to the pre-determined resource allocation with complete information.\\
\indent The centralized and the distributed AoI minimization schemes are compared in terms of overhead, implementation, and requirements. In particular, the centralized AoI minimization has an overhead of uplink communication request, while the distributed AoI minimization has an overhead of device-to-device communication. Simulation results show that the distributed AoI minimization achieves $40$-fold higher average instantaneous AoI than the centralized AoI minimization in a massive IoT, where the communication resources are highly limited and the devices only have partial information. In a less constrained IoT, simulation results show that the distributed AoI minimization achieves $8$-fold higher average instantaneous AoI than the centralized AoI minimization. Although the centralized AoI minimization outperforms the distributed AoI minimization in terms of average instantaneous AoI, the distributed approach may be more suitable for an IoT, because the centralized approach may not be practical or viable for an IoT. As such, our analysis clearly showcases the contrasts between the two solutions.\\
\indent The rest of this paper is organized as follows. Section II introduces the system model and the non-linear aging functions. Section III analyzes the AoI minimization with coexistence of linear and non-linear aging functions. Section IV presents the centralized and the distributed resource allocations in an IoT. Section V analyzes the simulation results, while Section VI draws conclusions.\vspace{-2mm}
\section{System Model}\label{sec:SM}
Consider the uplink of a wireless IoT system consisting of one BS serving $N$ IoT devices. The IoT devices can transmit their messages to the BS using the communication resources allocated by either a centralized or distributed resource allocation scheme. To transmit to the BS, the IoT devices use time-slotted orthogonal frequency-division multiple access (OFDMA). Here, $R$ time-frequency resource blocks (RBs) are allocated to the IoT devices at each time slot. If more than one IoT device use a given RB, none of the messages transmitted using the given RB can be successfully decoded, which leads to transmission failures. This implies that at most $R$ devices can transmit successfully to the BS at a time slot. In an IoT where the number of devices $N$ greatly outnumbers the number of RBs $R$, RB allocation is critical for the operation of IoT, and RB allocation can be done in a centralized or in a distributed way.\\
\indent Under a centralized resource allocation scheme, the BS allocates the RBs to the IoT devices such that a given RB is used by only one device. Therefore, duplicate RB usage will not occur when using a centralized resource allocation. However, centralized resource allocation incurs an overhead related to the need that the devices request their uplink communication resources via a random access channel (RACH) \cite{aoi3gpp}. Furthermore, the uplink communication resource request using RACH can fail resulting in transmission failure. In contrast, when using a distributed resource allocation scheme, the devices decide which RB to use autonomously without any intervention from the BS and without RACH. Although there is no overhead related to the need for requesting uplink communication resources, distributed resource allocation incurs an overhead related to the devices cooperating to avoid duplicate RB usage. Since there is no RACH request when performing distributed resource allocation, no RACH request failures will happen. However, the uplink transmission may fail because of a duplicate RB usage.\\
\indent For both centralized and distributed resource allocations, the uplink transmission can also fail because of the RB outage. The RB outage is based on the signal-to-noise ratio (SNR) such that the transmission is considered to be a failure if the SNR is less than a given threshold $\epsilon \geq 0$. We consider a stationary Rayleigh fading channel with additive white Gaussian noise (AWGN), such that the statistical properties of channel do not change over time. Therefore, the SNR outage probability is $\Pr\left(\sfrac{S^2}{\sigma^2} \leq \epsilon \right)$, where the received signal power $S^2$ is an exponentially distributed random variable and $\sigma^2$ is the variance of the AWGN. We assume that the IoT devices only know the distributional properties of the channel and the AWGN at the receiver. Furthermore, we assume that all devices transmit with the same transmit power as the devices do not know exact channel gain \cite{chenvr, iotalloc}, and \cite{aoisp1, aoisp2, aoisp3}. Since we consider a Rayleigh fading channel, the received signal power $S^2$ is exponentially distributed, and the devices know the mean $\lambda^{-1}$ of $S^2$. We assume that the transmit powers of all IoT devices are equal, and, thus, the mean of the received signal power will be the same for all devices. When an IoT device uses multiple RBs simultaneously, then the transmit power will be equally divided among those RBs. For instance, if an IoT device $i$ uses $R_{i,t}$ RBs simultaneously at time slot $t$, then the received signal power is $\sfrac{S^2}{R_{i,t}}$. For an IoT device $i$ using $R_{i,t}$ RBs simultaneously at time slot $t$, the outage probability $p_{i,t}$ for device $i$ at time slot $t$ will be: 
\begin{equation}
   p_{i,t} = \Pr\left(\frac{\sfrac{S^2}{R_{i,t}}}{\sigma^2} \leq \epsilon \right). \label{outp1}
\end{equation}
Since $S^2$ is exponentially distributed with mean $\lambda^{-1}$, $\frac{\sfrac{S^2}{R_{i,t}}}{\sigma^2}$ is exponentially distributed with mean $(R_{i,t} \sigma^2 \lambda)^{-1}$ for a given $R_{i,t}$ and a known $\sigma^2$. Moreover, the outage probability $p_{i,t}$ can be interpreted as a cumulative distribution function of an exponential distribution. Since an exponential random variable $S^2$ with mean $\lambda^{-1}$ has a cumulative distribution function of $\Pr\left(S^2 \leq \epsilon\right) = 1 - \exp\left(-\lambda \epsilon\right)$ for $\epsilon \geq 0$, the outage probability $p_{i,t}$ in \eqref{outp1} will be:
\begin{equation}
    p_{i,t} = 1 - \exp\left(-\left(R_{i,t} \sigma^2 \lambda\right)\epsilon\right). \label{outp2}
\end{equation}
For a successful uplink transmission, an RB must be used by only one device, and the SNR must be higher than a given threshold $\epsilon$.\\
\indent One prominent feature of an IoT is its massive scale. In particular, the number of IoT devices $N$ greatly outnumbers the number of RBs $R$. For an IoT scenario where $N > R$, the problem of RB allocation among the IoT devices becomes more challenging. Another prominent feature of an IoT is the heterogeneity among the IoT devices. The IoT devices are heterogeneous in terms of message types, transmission requirements, and packet content. A metric that can be used to determine which $R$ out of the $N$ devices will transmit and to quantify the freshness of the information in perspective of the destination is AoI. Furthermore, the heterogeneity among the IoT devices and their messages can be captured by extending the definition of AoI to include the quality of information and introducing non-linear aging functions.  
\subsection{Age of Information}
The AoI is a metric that quantifies the freshness of the information in the perspective of a destination \cite{bo2}, and the definition of the AoI is the time elapsed since the generation of a message that is most recently received at the BS. In prior art on the AoI, devices are commonly assumed to generate the messages at will \cite{gatwill2} and to update BS with a new message just in time \cite{jintime}. The generate-at-will model for AoI implies that IoT devices can have messages to transmit to the BS at any time, and the just-in-time model for AoI implies that IoT devices transmit new message to the BS immediately after the successful transmission of the current message. However, for an IoT, the devices are not always \textit{active} and do not always have the messages to transmit to the BS. In our model, a device has a message to transmit to the BS at a given time slot with an activation probability $v_a$. If a device transmits to the BS unsuccessfully at a given time slot, then the device retransmits the message immediately at following time slot without a random backoff time. For the distributed RB allocation, the proposed game is designed to give incentive to devices with low AoI to not transmit. The RACH phase of the centralized RB allocation can use the game formulated for distributed RB allocation to achieve a lower average instantaneous AoI. Meanwhile, the random backoff time used in \cite{aoibackoff} and \cite{backoff} does not consider different aging functions, when the backoff time is determined for each device. For instance, a long backoff time can be assigned to a device with higher AoI and exponential aging function, while a short backoff time can be assigned to a device with low AoI and linear aging function.\\ 
\indent Minimizing the AoI implies that the destination maintains fresh information from the source. However, minimizing the AoI is different from simply minimizing delay\cite{aoi1}. The AoI is measured by an aging function, and it is typically assumed that all devices have the same, linear aging function with a slope of $1$\cite{aoi1}. However, by using different aging functions for different messages, the AoI can naturally capture the inherent heterogeneity among IoT devices and messages. For instance, depending on the device type and the message characteristics, the aging function can be assigned appropriately. If the device is a simple sensor transmitting a time-insensitive update messages, the appropriate aging function is a linear aging function. On the other hand, if the device is an industrial monitoring sensor transmitting a time-sensitive status report, the appropriate aging function is an exponential aging function. By using different aging functions, the AoI captures both the freshness of information and the value of information \cite{nonlin1, nonlin2, nonlin3}.\\
\indent To model the heterogeneous messages, we consider the coexistence of linear aging function and exponential aging function. In particular, with a linear aging function, the AoI from IoT device $i$ at the beginning of time slot $t \in \mathbb{Z}_+$ is:
\begin{equation}
    a_i(t) = t - \delta_i(t), \label{af1}
\end{equation}
where $\delta_i(t)$ is time slot at which the most recent message from device $i$ received by the BS was generated. With an exponential aging function, the AoI from IoT device $i$ at the beginning of time slot $t \in \mathbb{Z}_+$ is:
\begin{equation}
    b_i(t) = 2^{t - \delta_i(t) - 1}. \label{af2}
\end{equation}
Although a specific linear and exponential aging functions are considered, our proposed centralized and distributed approaches to minimize the AoI are not limited to these aging functions only, and any type of aging function can be used.\\
\indent In a scenario where all IoT devices have the same linear aging function with slope $1$, the IoT devices and the BS can easily determine the AoI of the messages by simply counting the number of time slots passed since the most recently received message was generated \cite{aoi1}. However, in our system model where different aging functions coexist, the aging function of a given message is determined by the content of the message\cite{nonlin2}. For instance, the aging function of a message whose content is critical would be the exponential aging function $b_i(t)$, while the aging function of a message whose content is normal would be the linear aging function $a_i(t)$ \cite{nonlin1}. This implies that the BS cannot determine the aging function directly before receiving it, and, thus, the BS cannot compute the AoI of the messages. Therefore, the devices determine the aging function of their own message and compute the current AoI.\\
\indent To capture the heterogeneity among the IoT devices, the IoT devices are classified into different types based on the probabilistic properties of their messages. A typical IoT device would not always have a time-sensitive message to send to BS. Additionally, a device that usually sends time-insensitive messages may sometimes have a critical message to send. In our model, we consider two types of devices. Type $1$ devices are more likely to have linearly aging messages than exponentially aging messages. In other words, type $1$ devices have the messages with aging function $a_i(t)$ with probability $m_1$ and have messages with aging function $b_i(t)$ with probability $(1-m_1)$ with $1 > m_1 > 0.5$. Type $2$ devices are more likely to have exponentially aging messages than linearly aging messages. In other words, type $2$ devices have messages with aging function $b_i(t)$ with probability $m_2$ and have messages with aging function $a_i(t)$ with probability $(1 - m_2)$ with $1 > m_2 > 0.5$. We assume that the characteristics of devices types, such as $m_1$ and $m_2$, are known to the BS, but the BS does not know the type of a given device. Although having different types of messages realistically models the heterogeneity of the IoT devices, this makes the RB allocation more challenging, because the messages transmitted by a given device may have different aging functions. With non-linear aging functions and heterogeneous device types, the problem of AoI minimization is different from AoI minimization with only linear aging function and homogeneous devices. Hence, next we investigate the problem of non-linear AoI minimization with coexistence of different aging functions and heterogeneous devices.
\section{Non-linear AoI Minimization}
\indent To minimize the average instantaneous AoI, the devices with highest AoI are permitted to transmit to the BS, and the AoI of different devices are compared to decide which devices are allocated the RBs in a massive IoT. Without coexistence of different aging functions, comparing AoI from different devices to allocate the limited RBs is simple. If all devices have same linear aging function $a_i(t)$, then $a_i(\tau) > a_h(\tau)$ implies $a_i(\tau + \beta) > a_h(\tau + \beta)$ for any positive integers $\tau$ and $\beta$ given that devices $i$ and $h$ do not transmit successfully to the BS. Therefore, in a massive IoT with $N > R$, comparing current AoI with $t = \tau$ at time slot $\tau$ can be used to decide which $R$ out of $N$ devices transmit and to minimize the average instantaneous AoI. Furthermore, using current AoI with $t = \tau$ or using future AoI with $t = \tau+\beta$ at time slot $\tau$ is equivalent, when all devices have same linear aging function $a_i(t)$.\\ 
\indent When a linear aging function $a_i(t)$ and an exponential aging function $b_i(t)$ coexist, it is insufficient to only compare the AoI of different devices to minimize the average instantaneous AoI, and RBs must be allocated considering the aging functions to minimize the average instantaneous AoI. For instance, even if a device $i$ with $b_i(t)$ has lower AoI than a device $h$ with $a_h(t)$, device $i$ will eventually have higher AoI than device $h$, because $b_i(t)$ increases faster than $a_h(t)$. Therefore, the AoI comparison must take aging functions into account, and one way to consider the aging functions is to compare the future AoI.\\
\indent The optimization problem to minimize the average instantaneous AoI at time slot $\tau$ is:
\begin{align}
    \min_{\boldsymbol{n}}& \ \  \frac{1}{N}\left(\sum_{i \in \boldsymbol{n}} z_{i}'(\tau) + \sum_{i \not \in \boldsymbol{n}, i \in \boldsymbol{N}} z_{i}(\tau)\right)\\
    \textrm{s.t.}  &\ \ \boldsymbol{n} \subseteq \boldsymbol{N},\\
                   &\ \ |\boldsymbol{n}| = R,
\end{align}
where $\boldsymbol{N} = \{1, \cdots, N\}$ is set of all devices, $z_{i}'(\tau)$ is aging function of device $i$ such that $\delta_i(\tau) \neq \delta_i(\tau-1)$ is updated, and $z_{i}(\tau)$ is aging function of device $i$ such that $\delta_i(\tau) = \delta_i(\tau-1)$. $\boldsymbol{n}$ is a set of all devices allocated the RBs at $(\tau-1)$, while $\boldsymbol{N} - \boldsymbol{n}$ is a set of all devices not allocated the RBs at $(\tau-1)$. Next, for the problem of instantaneous AoI minimization, we prove that performing RB allocation based on the future AoI achieves a lower average instantaneous AoI than an alternative that is based on the current AoI.
\begin{proposition} \label{pp1}
\normalfont In a massive IoT where $N > R$, if there are different aging functions $a_i(t)$ and $b_i(t)$, comparing the future AoI with $t = \tau + \beta$ for some positive integer $\beta$ at time slot $\tau$ to determine the RB allocation achieves a lower average instantaneous AoI than comparing the current AoI with $t = \tau$.
\end{proposition}
\begin{proof} See Appendix \ref{app1}. \end{proof}
\indent From Proposition \ref{pp1}, we observe that allocating RBs to the devices with highest future AoI results in a lower overall average instantaneous AoI of all devices at all time slots than allocating RBs to the devices with highest current AoI. It is important to note that the current AoI at the time slot of successful transmission is used to compute the average instantaneous AoI, and the future AoI is only used to determine the RB allocation. In Proposition \ref{pp1}, the value of the positive integer $\beta$ in determining which future AoI to use for the RB allocation is a design parameter, and different values of $\beta$ have different effects. If higher values of $\beta$ are used, devices with exponentially aging messages are more likely to be allocated the RBs than devices with linearly aging messages. This implies that devices with exponentially aging messages are allocated the RBs even when their current AoI is low, while devices with linearly aging messages are not allocated the RBs even when their current AoI is high. Therefore, the exponentially aging messages are being transmitted before the linearly aging messages. In other words, with higher values of $\beta$, the average instantaneous AoI of IoT devices with exponentially aging messages is lower, while the average instantaneous AoI of IoT devices with linearly aging messages is higher. Furthermore, using higher values of $\beta$ for the future AoI does not necessarily achieve a lower average instantaneous AoI, and $\beta$ can be chosen depending on how much the exponentially aging messages are prioritized over the linearly aging messages. In our system model, the centralized and the distributed RB allocations use future AoI with $\beta = 1$.\\ 
\indent We consider that the IoT devices may have messages requiring multiple RBs to successfully transmit to the BS. When the messages take multiple RBs to successfully transmit, $\delta_i(t)$ in \eqref{af1} and \eqref{af2} will represent the time slot during which the most recent message from device $i$, which is fully received by BS, was generated. In this case, the devices must determine how to transmit the messages taking multiple RBs. In particular, if a device $i$ has a message that requires $n_i$ RBs to transmit, then the message may be transmitted simultaneously by using $n_i$ RBs at a given time slot, consecutively by using $1$ RB each time for $n_i$ time slots, or jointly by using both simultaneous and consecutive transmissions. The simultaneous transmission may complete the transmission at once reducing the AoI, but it has high $R_{i,t}$ and outage probability $p_{i,t}$ \eqref{outp2}. A consecutive transmission achieves its lowest outage probability $p_{i,t}$ with $R_{i,t} = 1$, but it takes the largest number of time slots to completely transmit increasing the AoI. Therefore, the problem of AoI minimization must consider the outage probability. The optimization problem to minimize the average instantaneous AoI for a device $i$ with a linearly aging message requiring $n_i$ RBs at time slot $\tau$ is:
\begin{align}
    \min_{\boldsymbol{R}}& \ \  a_i\left(\tau + \sum_{j = \tau}^{\tau + n_i - 1} (1 - p_{i,j})^{-1}\right), \label{opt}\\
    \textrm{s.t.}  &\ \ \sum\nolimits_{j = \tau}^{\tau + n_i - 1} R_{i,j} = n_i,\\
                   &\ \ 0 \leq R_{i,j} \leq R \ \forall \ j,
\end{align}
where $\boldsymbol{R} = [R_{i,\tau}, ..., R_{i, \tau+n_i-1}]$. Since $R_{i, \tau} \in \mathbb{Z}_+$, $|\boldsymbol{R}| = n_i$, and $n_i$ is typically small for the IoT devices \cite{iotshort}, the solution space of optimization problem \eqref{opt} is finite and small. Therefore, the optimization problem can be solved easily using any discrete optimization method, such as combinatorial optimization. In particular, for any $n_i$, the optimization problem \eqref{opt} can be mapped as a directed graph with $(1 - p_{i,j})^{-1}$ as weights to find a shortest path. When a device uses $R_{i,t}$ RBs simultaneously, $(1 - p_{i,t})$ is the probability of successful transmission given that duplicate RB selection did not occur. Taking $(1 - p_{i,t})$ as the success probability in a geometric distribution, the expected number of time slots needed for the successful transmission when using $R_{i,t}$ RBs simultaneously is $(1 - p_{i,t})^{-1}$, which is the mean of the geometric distribution. Therefore, $\sum_{j = \tau}^{\tau + n_i - 1} (1 - p_{i,t})^{-1}$ is the expected number of time slots needed to transmit a message requiring $n_i$ RBs. If a device $i$ has an exponentially aging message, $b_i(t)$ replaces $a_i(t)$ in \eqref{opt}. The solution to the optimization problem \eqref{opt} determines the number of RBs $R_{i,\tau}$ that a device $i$ should be allocated with at time slot $\tau$ so that the average instantaneous AoI is minimized. Furthermore, the solution to the optimization problem is used for centralized and distributed approaches for the RB allocation. 
\section{Resource Block Allocation}
In a massive IoT with $N > R$, RB allocation is a challenging problem especially given the high heterogeneity among IoT devices and the messages. RB allocation in an IoT can be done in a centralized way or in a distributed way. Moreover, RB allocation schemes can achieve a lower average instantaneous AoI by allocating the limited RBs to IoT devices with higher future AoI. Centralized RB allocation scheme is based on a priority scheduling improved with maximum likelihood to determine the aging functions and to learn the device types. The proposed distributed RB allocation scheme in Algorithm 1 is designed to enable only the devices with sufficiently high future AoI to transmit, and the proposed stochastic crowd avoidance algorithm is proved to converge to an NE of the formulated game.
\subsection{Centralized RB Allocation}
For centralized RB allocation, the BS allocates the RBs to the IoT devices. In a time slot. the centralized approach has two phases. The active devices request for the RB using RACH in the first phase. If an active device is allocated an RB, the active device transmits its message to the BS in the second phase. Although there will be no duplicate RB usage causing a transmission failure in the second phase, there may be RACH preamble collision causing an RB request failure in the first phase. Therefore, using centralized RB allocation scheme, a device fails to transmit to the BS because of the RACH preamble collision, the outage based on SNR, and the lack of RB allocation.\\
\indent We let $P$ be the number of RACH preambles and $N_t$ be the number of active devices at time slot $t$. The probability of the RACH preamble collision $c_t$ at time slot $t$ is:
\begin{equation}
    c_t = 1 - \left(\frac{P-1}{P}\right)^{N_t-1}, \label{prefail}
\end{equation}
which is the probability of more than one active device using a given RACH preamble. If a device fails to request for an RB at time $t$ with probability $c_t$, then this is equivalent to a transmission failure. However, if a device $i$ successfully requests for an RB at time $t$, IoT device $i$ sends the information about its current AoI $C_i$ and the necessary number of RBs $R_{i,t}$ at time slot $t$. After gathering the information from the devices, the BS determines the RB allocation based on the future AoI of devices to minimize the average instantaneous AoI.\\
\indent In the second phase, the BS allocates the RBs to the active devices using a priority scheduling based on the future AoI. The priority scheduling allocates the $R$ RBs to at most $R$ active devices with the highest future AoI, minimizing the average instantaneous AoI. If a device $i$ with high future AoI has $R_{i,t} > 1$, then IoT device $i$ may be allocated more than $1$ RB at time slot $t$. The problem in the second phase is determining the future AoI using the received current AoI because of the coexistence of different aging functions. In other words, the BS does not know the aging function of an active device $i$, and the BS must determine the aging function to compute the future AoI $F_i$, which is used for the priority scheduling scheme to achieve a lower average instantaneous AoI as shown in Proposition \ref{pp1}. However, the BS can determine the current aging function of an active device $i$ from $C_i$.\\
\indent The BS can determine that the aging function of an active device $i$ is $a_i(t)$, when $C_i$ cannot be derived using $b_i(t)$. The possible values of AoI using $a_i(t)$ are $\{1, 2, 3, 4, ...\}$, and the possible values of the AoI using $b_i(t)$ are $\{1, 2, 4, 8, ...\}$. Therefore, the values of AoI that are only possible using $a_i(t)$ are $\{3, 5, 6, 7, ...\}$. If the received current age $C_i$ from an active device $i$ is one of $\{3, 5, 6, 7, ...\}$, then the aging function is $a_i(t)$, and, thus, the future AoI $F_i$ of device $i$ is $C_i + 1$. The BS can also determine the aging functions of active devices by using regression. For an active device $i$, the BS can use the AoI from the most recently received uplink request from device $i$ and $C_i$ to determine if the aging function is $a_i(t)$ or $b_i(t)$. After determining the aging functions of active devices, the BS can compute the future AoI $F_i$ of active devices and learn the device types, which are used for the priority scheduling scheme with learning. The BS can use either of the two methods to accurately determine the current aging functions, but it is not always possible to use these methods. When both methods are not possible to use in a given time slot due to RACH preamble collision, the BS uses the expected value of $F_i$ for the priority scheduling scheme.\\
\indent The BS can compute the expected value of $F_i$ of an active device $i$ by learning the device type of device $i$. The BS can learn the type of a device $i$ by using the previous data from the instances that the BS was able to determine the aging function of messages from device $i$. In particular, the BS can use a maximum likelihood to determine the device types. We let $\mathcal{S}$ be a set of all device types and $\boldsymbol{O}_i$ be a vector of aging functions that a device $i$ had that the BS was able to determine exactly. Furthermore, we let $k_{i,f}$ be the number of times that the BS determined device $i$ to have aging function $f$. Assuming that the aging functions of a device $i$ are determined independently, the learned type $H_i \in \mathcal{S}$ of a device $i$ is:
\begin{align}
    H_i &= \argmax\limits_{s \in \mathcal{S}} \Pr(\boldsymbol{O}_i \mid s) = \argmax\limits_{s \in \mathcal{S}} \prod\limits_{f\in\mathcal{F}} {\Pr(f \mid s)}^{k_{i,f}},\\
    &= \argmax\limits_{s \in \mathcal{S}} \sum\limits_{f\in\mathcal{F}} {k_{i,f}}\ln(\Pr(f\mid s)), \label{ml}
\end{align}
where $\mathcal{F}$ is a set of all aging functions. For our model, the values of $\Pr(f \mid s)$ for any $f \in \mathcal{F}$ and $s \in \mathcal{S}$ are known. In particular, $\Pr(a_i \mid s = 1) = m_1$, $\Pr(b_i \mid s = 1) = 1 - m_1$, $\Pr(a_i \mid s = 2) = 1 - m_2$, and $\Pr(b_i \mid s = 2) = m_2$. Therefore, the maximum likelihood in \eqref{ml} can be solved by the BS directly.\\
\indent Once the device type of a device $i$ is learned, the expected future AoI $\mathbb{E}[F_i]$ of device $i$ can be computed. For our model, if $H_i = 1$, then the expected future AoI $\mathbb{E}[F_i\mid H_i = 1]$ is:
\begin{equation}
    \mathbb{E}[F_i \mid H_i = 1] = m_1(C_i + 1) + (1 - m_1)(2C_i).
\end{equation}
If $H_i = 2$, then the expected future AoI $\mathbb{E}[F_i\mid H_i = 2]$ is:
\begin{equation}
    \mathbb{E}[F_i \mid H_i = 2] = (1 - m_2)(C_i + 1) + m_2(2C_i).
\end{equation}
The expected value of $F_i$ is used for the priority scheduling scheme only if the exact value of $F_i$ cannot be determined.\\
\indent Priority scheduling determines the RB allocation among the active IoT devices based on the future AoI $F_i$. In a time slot $\tau$, the RBs are allocated first to the devices with the highest $F_i$. Furthermore, for a device $i$ with highest $F_i$ at time slot $\tau$, the number of RBs allocated to device $i$ is $R_{i, \tau}$. If some of the active devices have same $F_i$, then the RBs are allocated first to the devices whose type is more likely to have a faster aging function. For instance, if a device $i$ is device type $1$, a device $j$ is device type $2$, and $F_i = F_j$, then device $j$ has a priority over device $i$, because device $j$ is more likely to have exponentially aging messages. The RBs are allocated to the active devices until all RBs are allocated or all active devices are allocated the RBs.\\
\begin{algorithm}[t]
\caption{Priority scheduling based on future AoI at time slot $t$.}
\begin{algorithmic}[1]\label{centalg}\vspace{1mm}
\item[1 :] \hspace{0.0cm}  Receive values $C_i$ and $R_{i,t}$, and initialize $R$.
\item[2 :] \hspace{0.0cm}  Compute $F_i$ or $\mathbb{E}[F_i]$ for each $C_i$.
\item[3 :] \hspace{0.0cm}  {\bf for} $j = 1, 2, \cdots$
\item[4 :] \hspace{0.3cm}		$\mathcal{Z}_1 \leftarrow$ set of type $1$ devices with $j$-th highest value among $F_i$.
\item[5 :] \hspace{0.3cm}		$\mathcal{Z}_2 \leftarrow$ set of type $2$ devices with $j$-th highest value among $F_i$.
\item[6 :] \hspace{0.3cm}		{\bf for all $i \in \mathcal{Z}_2$}
\item[7 :] \hspace{0.6cm}		    {\bf if} $R > 0$, 
\item[8 :] \hspace{0.9cm}		        Allocate $\textrm{min}(R, R_{i,t})$ RBs to device $i$.
\item[9 :] \hspace{0.9cm}              $R \leftarrow R - \textrm{min}(R, R_{i,t})$. {\bf end if}
\item[10 :] \hspace{0.3cm}		{\bf end for}
\item[11 :] \hspace{0.3cm}		{\bf for all $i \in \mathcal{Z}_1$}
\item[12 :] \hspace{0.6cm}		    {\bf if} $R > 0$, 
\item[13 :] \hspace{0.9cm}		        Allocate $\textrm{min}(R, R_{i,t})$ RBs to device $i$.
\item[14 :] \hspace{0.9cm}              $R \leftarrow R - \textrm{min}(R, R_{i,t})$. {\bf end if}
\item[15 :] \hspace{0.3cm}		{\bf end for}
\item[16:] \hspace{0.0cm}  {\bf end for} 
\end{algorithmic}\vspace{1mm}
\end{algorithm}
\indent One of the major problems with priority scheduling is the infinite blocking of low-priority tasks. However, since the priority depends on the future AoI, the priority of low-priority messages increases with time. Therefore, regardless of the aging function or the device type, the messages will eventually be allocated the RBs to transmit to the BS. One of the limitations of centralized RB allocation scheme is the overhead related to the RB request via RACH. With the higher number of active devices $N_t$ at time $t$, the probability of the RACH preamble collision $c_t$ becomes significant, and, thus, more transmission failures occur. Furthermore, with centralized RB allocation scheme, a frequent communication between the devices and the BS is required, which may not be viable for the IoT devices \cite{nocent}. However, the main advantage of centralized RB allocation scheme is that the RBs are fully utilized. 
\subsection{Distributed RB Allocation}
Distributed RB allocation enables the devices to allocate the RBs in a self-organizing manner without any intervention from the BS. Since a duplicate RB selection results in transmission failures, the active devices must choose the RBs such that no other device is choosing the same RB. Furthermore, the distribution RB allocation can be modeled as one-to-one association between the RBs and the devices. The behavior of the devices wanting to choosing an RB alone can be formulated as a minority game \cite{iotkolkata}.\\
\indent A suitable minority game for distributed RB allocation in an IoT is the Kolkata paise restaurant (KPR) game \cite{kolkata}. The KPR game is a repeated game in which the customers simultaneous go to one of the restaurants, which can only serve one customer each. Additionally, the cost of going to a restaurant is the same for all restaurants, and a customer can only go to one restaurant at any given time. In the KPR game, the players are the customers, whose action in each iteration is to choose one of the restaurants. The payoff of a given player depends on the utility of the chosen restaurants and the number of players choosing the same restaurant. Given players, actions, and payoffs in a game, one important stable solution is an NE. A vector of actions is an NE if no player can achieve a higher payoff by a unilateral change of action. For the KPR game, the existence of an NE depends on the utilities of the restaurants \cite{kolkata}, and an NE is when all customers go to different restaurants and none of the customers have the utility of $0$. If an NE exists in the KPR game, it coincides with the socially optimal solution, which is when all restaurants are being utilized.\\
\indent The fundamental structure of the KPR game can be readily extended for our IoT model. The customers can be modeled as IoT devices, and the restaurants can be modeled as the RBs. Furthermore, the cost of using any of the RBs is same. However, there are significant differences between the KPR game and the IoT game. In the KPR game, the number of customers and the number of restaurants are the same, and each customer goes to one of the restaurants at every iteration. In the IoT game, the number of devices $N$ and the number of RBs $R$ may not be the same, and not all devices are active and need to use the RBs at each time slot. The most significant difference is the payoff in the case of duplicate RB selection. When multiple customers choose the same restaurant in the KPR game, one of those customers is randomly chosen to get the full payoff, while other customers with duplicate selection get a zero payoff. However, in the IoT game, all devices that choose the same RB get a zero payoff because of the transmission failures.\\ 
\indent For our AoI minimization, the players are the $N$ IoT devices, and their action is to transmit using the RBs or not to transmit. We let $\mathcal{R}$ be the set of $R$ RBs and let $x_i(t)$ be the action of device $i$ at time slot $t$. If $x_i(t) \in \mathcal{R}$, then device $i$ transmits using the RB in $x_i(t)$ at time slot $t$. If $x_i(t) = 0$, then device $i$ does not transmit at time slot $t$. at each time slot, the payoff of each device depends on the actions of all devices. If a device transmits successfully by using an RB alone, then the payoff is $\rho$. Furthermore, a successful transmission using any of the RBs has the same payoff of $\rho$. If a device transmits unsuccessfully due to a duplicate RB usage, then the payoff is $-\gamma$. The transmission failure has a negative payoff, because the energy is consumed for the transmission without success. Additionally, the transmission failure using any of the RBs has the same payoff of $-\gamma$, and $\rho$ and $\gamma$ are positive numbers such that $\rho > \gamma$.\\ 
\indent To minimize the AoI of the devices, active devices with lower AoI must not transmit, while the active devices with high AoI need to transmit. Using a distributed RB allocation scheme, the devices must know the AoI of other devices to determine if their own AoI is high enough to transmit. We assume that the active devices broadcast their own future AoI $F_i$ to other devices within the communication range $r_c$, and the communication resource for this broadcast is pre-allocated. Moreover, we assume that device-to-device communication links are orthogonal to the uplink communication as done in \cite{iotd2d1, iotd2d2, iotd2d3, iotd2d4, iotd2d5}. Similar to the overhead related to the RACH uplink request for centralized RB allocation scheme, the communication between the devices to share the AoI can be seen as the overhead for distributed RB allocation scheme. However, depending on $r_c$, the active devices may not know $F_i$ of all other active devices. We let $\alpha_i$ be the active status of a device $i$ such that $\alpha_i = 0$ implies that device $i$ is inactive and $\alpha_i = 1$ implies that device $i$ is active. We let $\boldsymbol{A}$ be a vector that captures the future AoI $F_i$ of all active devices, and $\boldsymbol{A}_i$ be a vector of the future AoI $F_i$ of the active devices within $r_c$ of an active device $i$. With $r_c$ sufficiently large, $|\boldsymbol{A}_i| = |\boldsymbol{A}|$ for all devices, where $|\boldsymbol{A}|$ is the cardinality of $\boldsymbol{A}$. For an active device $i$, if $F_i$ is higher than $\kappa$-th highest AoI in $\boldsymbol{A}_i$, then the active device $i$ transmits. $\kappa$ determines if $F_i$ is sufficiently higher than the future AoI of other active devices, and we let $\boldsymbol{A}_i(\kappa)$ be the $\kappa$-th highest AoI in $\boldsymbol{A}_i$. Moreover, $\kappa$ should ensure that the number of transmitting devices $T_t$ at time slot $t$ is equal to $R$ such that all transmitting devices can be allocated with the RB. If $T_t$ is higher than $R$, then there are at least two active devices with transmission failure, which causes the average instantaneous AoI to increase. If $T_t$ is less than $R$, then there are some of the RBs not used by any of the active devices, which may cause the average instantaneous AoI to increase. However, $T_t = R$ is difficult to achieve when $r_c$ is not sufficiently large and $|\boldsymbol{A}_i| < |\boldsymbol{A}|$.\\ 
\indent When $|\boldsymbol{A}_i| = |\boldsymbol{A}|$, the devices have full information on the future AoI $F_i$ of the active devices. Moreover, the devices know all active devices. Since the devices have full information of $F_i$, an active device $i$ transmits if $F_i \geq \boldsymbol{A}_i(R)$, and an active device $i$ does not transmit if $F_i < \boldsymbol{A}_i(R)$. Therefore, $\kappa = R$ when $|\boldsymbol{A}_i| = |\boldsymbol{A}|$ for all $i$. This ensures that $R$ active devices transmit, while $|\boldsymbol{A}| - R$ devices do not transmit. Therefore, the number of transmitting devices $T_t$ at time slot $t$ is equal to $R$. To formulate this decision to transmit or not to transmit into the payoff, the payoff $y_{\textrm{full}}$ when an active device $i$ does not transmit is: \vspace{-2mm}
\begin{multline}
    y_{\textrm{full}}(\boldsymbol{A}_i) = (\rho + \eta) \theta_{\mathbb{R}_{+}}(\boldsymbol{A}_i(R) - F_i-\eta)\\ 
    - (\gamma + \eta) \theta_{\mathbb{R}_{+}}(F_i - \boldsymbol{A}_i(R)), \label{payoff1}
\end{multline}
where $\eta$ is a real number in $(0, 1)$ and $\theta_{\mathbb{R}_{+}}$ is an indicator function such that:
\begin{equation}
    \theta_{\mathbb{R}_{+}}(x)=\left\{
                \begin{array}{ll}
                  1 \hfill &\text{if} \ x \in [0, \infty),\\
                  0 \hfill &\text{if} \ x \not\in [0, \infty).
                \end{array}
              \right.
\end{equation}
The function of $\eta$ in payoff functions is to ensure that the payoff functions are used as intended and only appropriate indicator function is activated. It is important to note that $y_{\textrm{full}}(\boldsymbol{A}_i)$ does not depend on actions of the players, because the decision to transmit or not to transmit only depends on future AoI. Given the payoff in \eqref{payoff1}, the $R$ active devices with $F_i \geq \boldsymbol{A}_i(R)$ transmit, because their payoff of not transmitting is $-(\gamma + \eta)$. The $|\boldsymbol{A}| - R$ active devices with $F_i < \boldsymbol{A}_i(R)$ do not transmit, because their payoff of not transmitting is $\rho + \eta$. Therefore, IoT devices with $R$ highest future AoI transmit, while other devices do not transmit.\\
\indent In a more realistic scenario where IoT devices do not have full information of $F_i$, $|\boldsymbol{A}_i| < |\boldsymbol{A}|$, and the devices do not know all active devices. It is difficult to make only $R$ active devices to transmit at time slot $t$, and, thus, $\kappa$ is designed to make $T_t \approx R$. For the case of $|\boldsymbol{A}_i| < |\boldsymbol{A}|$, $\kappa$ is:
\begin{equation}
    \kappa = \left\lceil\frac{R}{|\boldsymbol{A}|} |\boldsymbol{A}_i|\right\rceil, \label{kappa1}
\end{equation}
where $\kappa = R$ in the case of $\boldsymbol{A} = \boldsymbol{A}_i$. However, with partial information, the active devices do not know $\boldsymbol{A}$ in \eqref{kappa1}. $|\boldsymbol{A}|$ is the number of active devices, and the expected number of newly active devices is $Nv_a$. However, with previously active devices yet to transmit successfully, $|\boldsymbol{A}|$ is typically greater than $Nv_a$. Therefore, in \eqref{kappa1}, $|\boldsymbol{A}|$ can be estimated with $N v_a \zeta$, where $\zeta$ is a design parameter to consider the number of previously active devices yet to transmit successfully. With higher values of $\zeta$, the number of transmitting devices $T_t$ at time slot $t$ is smaller, while $T_t$ is bigger with smaller values of $\zeta$. Therefore, with approximation for $|\boldsymbol{A}|$, $\kappa$ is:
\begin{equation}
    \kappa = \left\lceil\frac{R}{N v_a \zeta} |\boldsymbol{A}_i|\right\rceil. \label{kappa2}
\end{equation}
For an active device $i$, $\kappa$ in \eqref{kappa2} approximates if $F_i$ is sufficiently higher than the future AoI of other active devices based on the percentile of $F_i$ on the known vector of future AoI $\boldsymbol{A}_i$. With $|\boldsymbol{A}_i| < |\boldsymbol{A}|$, the payoff $y_{\textrm{act}}$ when an active device $i$ does not transmit is:\vspace{-2mm}
\begin{multline}
    y_{\textrm{act}}(\boldsymbol{A}_i) = (\rho + \eta) \theta_{\mathbb{R}_{+}}(\boldsymbol{A}_i(\kappa) - F_i-\eta)\\
        - (\gamma + \eta) \theta_{\mathbb{R}_{+}}(F_i - \boldsymbol{A}_i(\kappa)). \label{payoff2}
\end{multline}
It is important to note that $y_{\textrm{act}}(\boldsymbol{A}_i)$ is equal to $y_{\textrm{full}}(\boldsymbol{A}_i)$, when the active devices have full information with $|\boldsymbol{A}_i| = |\boldsymbol{A}|$ and $\kappa = R$. Similar to $y_{\textrm{full}}(\boldsymbol{A}_i)$ \eqref{payoff1}, with the payoff $y_{\textrm{act}}(\boldsymbol{A}_i)$, the active devices with sufficiently high $F_i$ satisfying $F_i \geq \boldsymbol{A}_i(\kappa)$ transmit, while the active devices with $F_i$ such that $F_i < \boldsymbol{A}_i(\kappa)$ do not transmit. Moreover, $y_{\textrm{act}}(\boldsymbol{A}_i)$ also does not depend on actions of the players.\\
\indent In the IoT game, the payoff needs to consider the inactive devices. The inactive devices with $\alpha_i = 0$ do not transmit as they do not have messages to transmit. The payoff $y_{\textrm{nt}}$ when an device $i$ does not transmit is: \vspace{-1mm}
\begin{equation}
    y_{\textrm{nt}}(\boldsymbol{A}_i, \alpha_i) = (1 - \alpha_i)(\rho + \eta) + \alpha_i y_{\textrm{act}}(\boldsymbol{A}_i). \vspace{-1mm}
\end{equation} 
With payoff $y_{\textrm{nt}}(\boldsymbol{A}_i, \alpha_i)$ for not transmitting, the inactive devices with $\alpha_i = 0$ do not transmit as they get payoff of $\rho + \eta$. The active devices with $\alpha_i = 1$ decide to transmit or to not transmit based on $\kappa$ and future AoI. We let $\boldsymbol{x}(t) = [x_1(t), x_2(t), \cdots, x_N(t)]$ be a vector of all actions of $N$ devices at time slot $t$. For a given $\boldsymbol{x}(t)$, the payoff function $y_i(\boldsymbol{x}(t), \boldsymbol{A}_i, \alpha_i)$ for a device $i$ at time slot $t$ is: \vspace{-2mm}
\begin{multline}
    y_i(\boldsymbol{x}(t), \boldsymbol{A}_i, \alpha_i)\\ =\left\{
                \begin{array}{ll}
                  \rho \hfill &\text{if} \ x_i(t) \neq x_j(t) \ \forall \ j \neq i, x_i(t) \neq 0,\\
                  -\gamma \hfill &\text{if} \ \exists \ j \neq i \ \text{s.t.} \ x_j(t) = x_i(t) \neq 0,\\
                  y_{\textrm{nt}}(\boldsymbol{A}_i, \alpha_i) \hfill &\text{if} \ x_i(t) = 0.\\
                \end{array}
              \right. \label{payoff3}\vspace{-2mm}
\end{multline} 
In a simple game where $N = R = 2$ with $v_a = 1$, the payoffs of two devices at each time slot is summarized in Table \ref{game1}. NE in this IoT game is $\boldsymbol{x}(t) = [1, 2]$ or $\boldsymbol{x}(t) = [2, 1]$. For simple IoT game, NE is when two devices choose different RBs and get the payoff of $\rho$. If one device deviates from NE and the other device does not deviate from NE, the device deviating from NE gets the lower payoff of $-\gamma$ or $-(\gamma + \eta)$. Furthermore, in the IoT game, NE implies that a duplicate RB selection does not occur. An NE for a more general case of the IoT game can be found with certain conditions. 

\begin{table}[t]
\centering \vspace{0mm}
\caption{IoT game with $N = R = 2$.}
\label{game1}
\begin{tabular}{rccc}
                                  & $x_1(t) = 1$                                    & $x_1(t)=2$                                      & $x_1(t) = 0$                                              \\ \cline{2-4} 
\multicolumn{1}{r|}{$x_2(t) = 1$} & \multicolumn{1}{c|}{$(-\gamma, -\gamma)$}       & \multicolumn{1}{c|}{$(\rho, \rho)$}             & \multicolumn{1}{c|}{$(-(\gamma + \eta), \rho)$}           \\ \cline{2-4} 
\multicolumn{1}{r|}{$x_2(t) = 2$} & \multicolumn{1}{c|}{$(\rho, \rho)$}             & \multicolumn{1}{c|}{$(-\gamma, -\gamma)$}       & \multicolumn{1}{c|}{$(-(\gamma + \eta), \rho)$}           \\ \cline{2-4} 
\multicolumn{1}{r|}{$x_2(t) = 0$} & \multicolumn{1}{c|}{$(\rho, -(\gamma + \eta))$} & \multicolumn{1}{c|}{$(\rho, -(\gamma + \eta))$} & \multicolumn{1}{c|}{$(-(\gamma+\eta), -(\gamma + \eta))$} \\ \cline{2-4} 
\end{tabular}\vspace{-4mm}
\end{table}

\begin{theorem}\label{thm1}
For the IoT game with $N$ players with action $x_i(t) \in \{0, \mathcal{R}\}$, payoff function $y_i(\boldsymbol{x}(t))$, and $|\boldsymbol{A}_i| = |\boldsymbol{A}|$ for all $i$, any vector of actions $\boldsymbol{x}(t)$ such that at most $R$ active devices with $F_i \geq \boldsymbol{A}_i(R)$ transmit, the rest of the devices do not transmit, and each of the RBs is used by at most one device is an NE. 
\end{theorem}
\begin{proof} See Appendix \ref{app2}. \end{proof}
\indent There are many sets of actions that satisfy the conditions described in Theorem \ref{thm1}, and, thus, NE in the IoT game is not unique. For instance, when the number of devices is equal to the number of RBs with $v_a = 1$ in an IoT, an NE is when each RB is used by one device, and, thus, the number of NEs in that particular IoT game is $N!$. There are many NEs, because the payoff of successful transmission does not depend on which RB is used. \emph{Although there are many NEs, the expected payoffs of devices at any given NE are the same, and, thus, one of those NEs is chosen with distributed RB allocation algorithm discussed in Section \ref{secsca}.} Similar to the simple IoT game in Table \ref{game1}, an NE for our IoT game implies that the number of transmitting devices $T_t$ is equal to the number of RBs $R$ and that a duplicate RB selection does not occur. Furthermore, at an NE, a transmitting device has one of the $R$ highest AoI, and a device that is not transmitting is inactive or has low AoI. This is because the payoff function with $|\boldsymbol{A}_i| = |\boldsymbol{A}|$ is designed to only allow the messages having the $R$ highest AoI to transmit. Therefore, the convergence of a distributed RB allocation algorithm to an NE reduces the average instantaneous AoI. However, when the devices only have partial information such that $|\boldsymbol{A}_i| < |\boldsymbol{A}|$, $\boldsymbol{x}(t)$ described in the Theorem \ref{thm1} is not necessarily an NE, because $\kappa$ is not necessarily equal to $R$.\\
\indent In addition to the NE, another solution concept is a socially optimal solution in which the overall payoff of an IoT game is maximized. In other words, a vector of actions is a socially optimal solution when the sum of payoffs of all devices is maximized. In an IoT game, a socially optimal solution is when the devices with highest future AoI $F_i$ fully utilize RBs without any duplicate RB selection. Therefore, similar to an NE in KPR game \cite{kolkata}, an NE in IoT game coincides with a socially optimal solution. A performance metric that can be used to describe an NE and a socially optimal solution is a service rate $s_r$, which is the percentage of RBs that are used by one device. NE in Theorem \ref{thm1} has a service rate of $1$, which implies that all RBs are used by one device, or the highest possible service rate of $\sfrac{T_t}{R}$.\\
\indent With the number of transmitting devices $T_t$ approximately equal to $R$ via IoT game design, a distributed RB allocation algorithm is necessary to enable the transmitting devices to share RBs autonomously. Moreover, with existence of a socially optimal NE in our IoT game, the convergence of a distributed RB allocation algorithm to an NE is crucial in minimizing the average instantaneous AoI. Therefore, to evaluate different RB allocation algorithms, the convergence to an NE and the service rates are analyzed. The service rate is an important metric to determine convergence to an NE, because a vector of actions must achieve service rate of $1$ or highest service rate to be an NE. Furthermore, the service rate is an important performance metric for the AoI, because the higher service rate implies that more devices are transmitting successfully at each time slot, reducing the average instantaneous AoI. Therefore, a distributed RB allocation algorithm must achieve a service rate of $1$ or increase the service rate as high as possible, because a high service rate is required to achieve a low average instantaneous AoI.

\subsubsection{Stochastic Crowd Avoidance} \label{secsca}
We propose a stochastic crowd avoidance (SCA) algorithm that enables the devices to avoid using the RBs that are used by many devices stochastically and to choose an RB for a successful transmission, as shown in Algorithm 2. For the SCA algorithm, the devices need to share more information in addition to their $F_i$ and can perform a channel sensing to determine the RBs that were not used at a previous time slot \cite{aoibackoff, aoicsma, aoisleep}, and \cite{aoiinst}. At time slot $t$, the devices that transmitted at the time slot $t-1$ share their previous actions $x_i(t-1) \in \mathcal{R}$ and the previous payoff $y_i(\boldsymbol{x}(t-1))$ of their transmission. We let $\boldsymbol{X}_i$ be the vector of actions of the transmitting devices at time slot $t-1$ that a device $i$ knows and $\boldsymbol{P}_i$ be the vector of payoffs of the transmitting devices at time slot $t-1$ that a device $i$ knows. Furthermore, we let $\boldsymbol{X}_i(x)$ with $x \in \mathcal{R}$ be the number of $x$ in $\boldsymbol{X}_i$. In other words, $\boldsymbol{X}_i(x)$ is the number of devices that chose the RB $x$ at time slot $t-1$ that a device $i$ knows. Learning from $\boldsymbol{X}_i$ and $\boldsymbol{P}_i$, the proposed SCA algorithm enables a transmitting device $i$ at time slot $t$ to not use the RBs that are being used successfully by other devices and to avoid using the contended RBs stochastically.\\
\begin{algorithm}[t]
\caption{SCA for device $i$ at time $t$.}
\begin{algorithmic}[1]\label{distalg}\vspace{1mm}
\item[1 :] \hspace{0.0cm}  Receive $\boldsymbol{X}_i$, $\boldsymbol{P}_i$, $\boldsymbol{A}_i$, and $\mathcal{L}$.
\item[2 :] \hspace{0.0cm}  {\bf if} $x_i(t-1) \in \mathcal{R}$, $y_i(\boldsymbol{x}(t-1)) = \rho$, and $\alpha_i = 1$,
\item[3 :] \hspace{0.3cm}       $x_i(t) \leftarrow x_i(t-1)$.
\item[4 :] \hspace{0.0cm}  {\bf else if} $x_i(t-1) \in \mathcal{R}$, $y_i(\boldsymbol{x}(t-1)) = \rho$, and $\alpha_i = 0$,
\item[5 :] \hspace{0.3cm}		$j \leftarrow$ one neighboring device chosen with $\frac{F_j}{\sum_{F_h \in \boldsymbol{A}_i} F_h}$.
\item[6 :] \hspace{0.3cm}		$x_j(t) \leftarrow x_i(t-1)$.
\item[7 :] \hspace{0.0cm}  {\bf else if} $x_i(t-1) \in \mathcal{R}$ and $y_i(\boldsymbol{x}(t-1)) = -\gamma$,
\item[8 :] \hspace{0.3cm}		$z \leftarrow 1$ with probability $\boldsymbol{X}_i(x_i(t-1))^{-1}$.
\item[9 :] \hspace{0.3cm}       {\bf if} $z = 1$, $x_i(t) \leftarrow x_i(t-1)$.
\item[10 :] \hspace{0.3cm}		{\bf else} $x_i(t) \leftarrow$ randomly chosen from $\mathcal{L}$. {\bf end if}.  
\item[11 :] \hspace{0.0cm} {\bf else if} $x_i(t-1) = 0$,
\item[12 :] \hspace{0.3cm}		$x_i(t) \leftarrow$ randomly chosen from $\mathcal{L}$.
\item[13 :] \hspace{0.0cm} {\bf end if.}
\end{algorithmic}
\end{algorithm}
\indent Using SCA algorithm, at time slot $t$, a transmitting device $i$ determines its RB usage based on $x_i(t-1)$ and $y_i(\boldsymbol{x}(t-1))$. If the transmission at time slot $t-1$ is successful such that $x_i(t-1) \in \mathcal{R}$ and $y_i(\boldsymbol{x}(t-1)) = \rho$, then transmitting device $i$ uses the same RB $x_i(t) = x_i(t-1)$. If the transmission at time slot $t-1$ is unsuccessful such that $x_i(t-1) \in \mathcal{R}$ and $y_i(\boldsymbol{x}(t-1)) = -\gamma$, then transmitting device $i$ uses the same RB $x_i(t) = x_i(t-1)$ with probability $\boldsymbol{X}_i(x_i(t-1))^{-1}$ or chooses an RB from a set $\mathcal{L}$ uniformly randomly. $\mathcal{L}$ is a set of the RBs that were not used by any of the device at time slot $t-1$ determined with channel sensing. If there is no transmission at time slot $t-1$ such that $x_i(t-1) = 0$, then the transmitting device $i$ chooses an RB from $\mathcal{L}$ uniformly randomly. In the case where $v_a < 1$, a device that transmitted successfully at time slot $t-1$ may no longer be active at time slot $t$. In this case, a device $i$ with $\alpha_i = 0$ and $y_i(\boldsymbol{x}(t-1)) = \rho$ chooses a neighboring device $j$ with $\alpha_j = 1$ with a probability proportional to $F_j$, such that the probability is $\sfrac{F_j}{\sum_{F_h \in \boldsymbol{A}_i} F_h}$.\\
\indent In our SCA algorithm, the RB that is used successfully at time slot $t-1$ is also used successfully at time slot $t$ if $v_a = 1$ or if there is an active neighboring device. Moreover, with $r_c$ sufficiently large, the expected number of transmitting devices using an RB that is used by more than one device at the previous time slot $t-1$ is $1$ at the current time slot $t$. This is because the probability of choosing the same RB after a transmission failure is $\sfrac{1}{\boldsymbol{X}_i(x_i(t-1))}$. The devices avoid to use the same RB stochastically even after a transmission failure, because the strict crowd avoidance causes the crowding in other RBs resulting in more transmission failures. Furthermore, when a device $i$ chooses some other RB such that $x_i(t) \neq x_i(t-1)$, device $i$ chooses from a set of RBs $\mathcal{L}$ that were not used at time slot $t-1$. This is to avoid using the RBs that are either used successfully or crowded, both of which cause the transmission failures. With SCA algorithm design to avoid duplicate RB selection, next we prove that the proposed SCA algorithm converges to an NE under certain IoT system parameters.
\begin{theorem}\label{thm2}
When $N$ devices are always active with full information and use $1$ out of $N$ RBs to transmit with negligible outage probability $p_{i,t}$ at each time slot, the vector of actions $\boldsymbol{x}(t)$ converges to an NE using SCA.
\end{theorem}
\begin{proof} See Appendix \ref{app3}. \end{proof}
Under the conditions in Theorem \ref{thm2}, $\boldsymbol{x}(t)$ converges to an NE, and this implies that the service rate increases to $1$. In general, SCA algorithm increases the service rate, because an RB, which is used by $1$ device at previous time slot $t-1$, is still used by $1$ device at current time slot $t$. Therefore, SCA algorithm is effective in reducing the average instantaneous AoI. However, SCA is susceptible to the high outage probability $p_{i,t}$ based on the SNR. This is because SCA cannot distinguish the transmission failure due to the duplicate RB selection and the transmission failure due to the outage based on the SNR. Furthermore, with partial information $|\boldsymbol{A}_i| < |\boldsymbol{A}|$, the devices also have partial information of $\boldsymbol{X}_i$ and $\boldsymbol{P}_i$, and, thus, the devices cannot choose $x_i(t)$ accurately in the case of transmission failure.\\
\indent In a massive IoT with $N > R$ and partial information, the vector of actions $\boldsymbol{x}(t)$ does not converge to an NE using SCA, because it is not possible to achieve service rate of $1$. The service rate cannot be $1$, because the transmitting devices are changing every time slot and a duplicate RB selection is inevitable. Since the service rate cannot be $1$, the average instantaneous AoI in a massive IoT is higher than the average instantaneous AoI in an ideal IoT described in Theorem \ref{thm2}. However, in a massive IoT, the proposed SCA algorithm still enables the transmitting devices to stochastically avoid duplicate RB selection with available information. Therefore, the proposed SCA algorithm still increases the service rate and reduces the average instantaneous AoI. However, in that case, it does not reach an NE but rather a sub-optimal, heuristic solution. To evaluate performance of the proposed SCA algorithm, next we study a random RB selection for distributed RB allocation scheme.
\subsubsection{Random RB Selection} 
One way for the transmitting devices to determine their RB usage is via random selection. In other words, the actions $x_i(t)$ of the transmitting devices are chosen uniformly random in $\mathcal{R}$. This random RB selection is used as a baseline. Even with $|\boldsymbol{A}_i| = |\boldsymbol{A}|$ for all $i$, the random RB selection is highly unlikely to achieve $\boldsymbol{x}(t)$ such that each of the RBs is used by at most one device, which is the requirement of $\boldsymbol{x}(t)$ to be an NE. Furthermore, the service rate $s_r$ using random RB selection for the transmitting devices is low. 
\begin{proposition}
\normalfont At a time slot $t$, the service rate $s_r$ with $T_t$ transmitting devices using random RB selection is:
\begin{equation}
    s_r = \frac{T_t}{R} \left(\frac{R-1}{R}\right)^{T_t - 1},\label{srr}
\end{equation}
and, for a massive IoT with $N$ increasing to infinity, the service rate $s_r$ is:
\begin{equation}
    \lim_{N \rightarrow \infty} s_r = \frac{T_t}{R-1} \exp\left(\frac{-T_t}{R}\right) \label{srriot}
\end{equation}
\end{proposition}
\begin{proof} See Appendix \ref{app4}. \end{proof}
Under the conditions in Theorem \ref{thm2}, when $N$ devices are always active and use $1$ out of $N$ RBs, the number of transmitting devices $T_t$ is equal to $N$. In this case, the service rate using random RB selection is always less than $1$ even with $N = R \geq 2$. On the other hand, for a massive IoT with $N > R$, the service rate using random RB selection exponentially decreases to $0$ as $N$ increases. Therefore, the random RB selection is not suitable for a massive IoT in which the number of devices $N$ outnumbers the number of RBs $R$, Furthermore, with low service rate, the probability of a successful transmission is low for a device, and, thus, the average instantaneous AoI is high. 
\section{Simulation Results and Analysis} 
For our simulations, we consider a rectangular area with width $w$ and length $l$ within which the $N$ devices are deployed following a Poisson point process. We let $w = l = 10$ m and $R = 50$ with a $10$ MHz frequency band \cite{50rblte}, while the number of devices $N$ will be varied for analysis. We choose a time slot duration of $1$ ms \cite{1msts} and expected value of SNR of $20$ dB with $\lambda = .1$ and $\sigma^2 = 0.001$. To vary outage probability $p_{i,t}$, different values of $\epsilon$ are used. Moreover, a device is assumed to be of type $1$ with probability $0.6$ with $m_1 = 0.75$, while a device is assumed to be of type $2$ with probability $0.4$ with $m_2 = 0.75$. The average of the current AoI $C_i$ at the time slot of successful transmission is the performance metric for different RB allocation schemes, and their performances are analyzed with varying $v_a$ and $p_{i,t}$.\\
\indent For the centralized RB allocation scheme, the number of RACH preambles $P$ for the uplink transmission request is $64$ \cite{rach}. Three different kinds of priority scheduling are analyzed. Priority scheduling without learning \cite{aoischedule1, aoischedule3, multihop, adhoc} does not learn the device types, and, hence, this scheme is used for baseline comparison. The proposed priority scheduling with learning learns the device types using maximum likelihood and information on $F_i$. Priority scheduling with full information assumes that the BS always knows the types of all devices, and, hence, this scheme is used for optimal performance comparison. All three priority scheduling algorithms are analyzed with $N = 500$, while varying $v_a$ and $p_{i,t}$.

\begin{figure}[t]       
\centering
\includegraphics[width = 9cm]{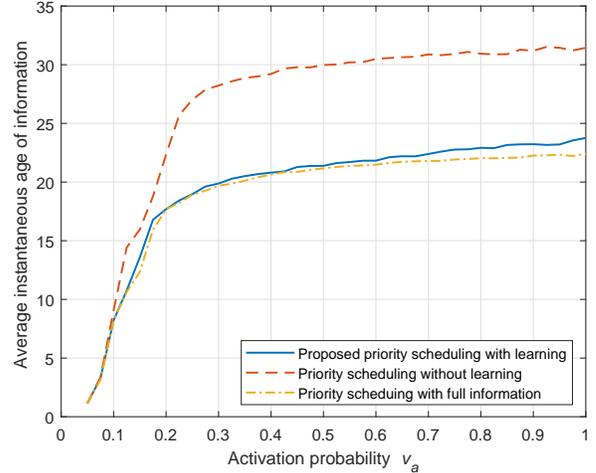}\vspace{-.2cm}
\caption{Average instantaneous AoI using centralized RB allocation schemes while varying $v_a$.}\vspace{-.5cm}
\label{cpa}
\end{figure}

Fig. \ref{cpa} shows the average instantaneous AoI of the devices using centralized RB allocation schemes for different values of the activation probability $v_a$ with $p_{i,t} = 0.01$ and $\epsilon = 1$. The average instantaneous AoI for the no learning case quickly increases to $27.88$ and then increases slowly above $30$, while the average instantaneous AoI for both learning and full information quickly increases to $19.62$ and then increases slowly above $20$. As $v_a$ increases, the average instantaneous AoI increases for all priority scheduling algorithms, because more devices are transmitting and RACH preamble collisions are more likely to occur. However, for $v_a > 0.2$, the average instantaneous AoI flattens and increases at a much slower rate with increasing $v_a$ for all priority scheduling algorithms, because all RBs are fully saturated. Moreover, there is a significant difference between priority scheduling with learning and without learning. After the average instantaneous AoI flattens, the difference of the average instantaneous AoI between priority scheduling with learning and without learning is constantly about $8$. This implies that the proposed priority scheduling scheme with learning achieves about $26.7\%$ lower average instantaneous AoI when compared to simple priority scheduling scheme. Hence, learning the device types is important to decrease the average instantaneous AoI for priority scheduling. However, there is an insignificant difference between priority scheduling with learning and with full information, which implies that the learning is effective in learning the device types. 

\begin{figure}[t]       
\centering
\includegraphics[width = 9cm]{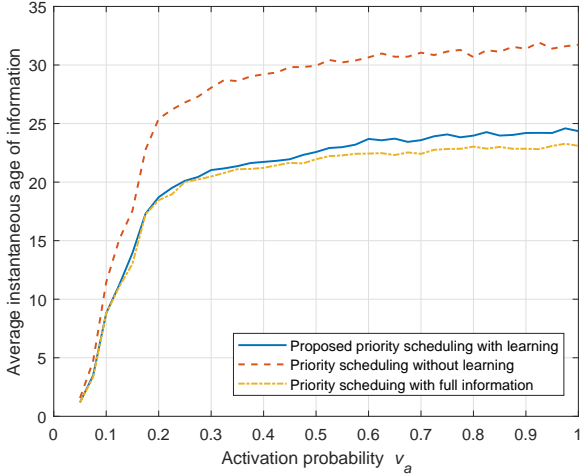}\vspace{-.2cm}
\caption{Average instantaneous AoI using centralized RB allocation schemes with different transmit powers while varying $v_a$.}\vspace{-.5cm}
\label{cpaxp}
\end{figure}

Fig. \ref{cpaxp} shows the average instantaneous AoI of the devices with different transmit powers using centralized RB allocation schemes for different values of the activation probability $v_a$ with $p_{i,t} = 0.01$ and $\epsilon = 1$. The devices have different transmit powers such that their SNR values are uniformly distributed random variables from $17.0$ dB to $21.8$ dB, and the only difference between Fig. \ref{cpa} and Fig. \ref{cpaxp} is assumption on the transmit powers of devices. The overall trends of the average instantaneous AoI for all three centralized RB allocation schemes are similar to the trends shown in Fig. \ref{cpa}. However, a notable difference is the average instantaneous AoI after the curves flatten. After the average instantaneous AoI flattens, the average instantaneous AoI increases slowly above $30$ for priority scheduling without learning, increases slowly to $25$ for priority scheduling with learning, and increases slowly to $23.5$ for priority scheduling with full information. With different SNR values randomly assigned to the devices, some devices have higher outage probability and other devices have lower outage probability compared to the devices in Fig. \ref{cpa}. With exponentially aging messages, an increase in the average instantaneous AoI from devices with higher outage probability outweighs a decrease in the average instantaneous AoI from devices with lower outage probability. Therefore, there is a slight increase in average instantaneous AoI when the devices have different transmit powers. However, similar to Fig. \ref{cpa}, the learning can still effectively decrease the average instantaneous AoI, and priority scheduling with learning performs similar to priority scheduling with full information even when the devices have different transmit powers.

\begin{figure}[t]       
\centering
\includegraphics[width = 9cm]{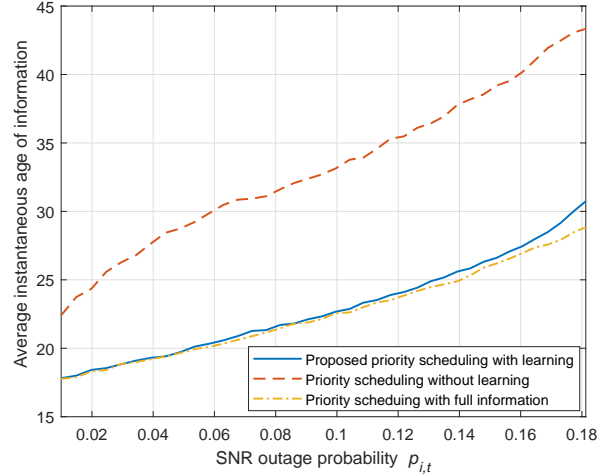}\vspace{-.2cm}
\caption{Average instantaneous AoI using centralized RB allocation schemes while varying $p_{i,t}$.}\vspace{-.5cm}
\label{ceps}
\end{figure}

Fig. \ref{ceps} shows the average instantaneous AoI of the devices using centralized RB allocation schemes for different values of the SNR outage probability $p_{i,t}$ with $v_a = 0.2$ and varying $\epsilon$ from $1$ to $20$. Unlike in Fig. \ref{cpa}, the average instantaneous AoI does not flatten and increases at about the same rate. With $p_{i,t} = 0.02$, the average instantaneous AoI without learning is $24.34$, and the average instantaneous AoI with learning is $18.41$. With $p_{i,t} = 0.16$, the average instantaneous AoI without learning is $40.17$, and the average instantaneous AoI with learning is $27.45$. Therefore, the difference between the average instantaneous AoI with learning and without learning increases as $p_{i,t}$ increases. Furthermore, with high $p_{i,t}$, the proposed priority scheduling scheme with learning achieves about $31.7\%$ lower average instantaneous AoI when compared to simple priority scheduling scheme. For a higher $p_{i,t}$ and more frequent transmission failures, the AoI of the exponentially aging messages becomes much higher than the AoI of the linearly aging messages. In this case, the learning scheme, which enables the BS to accurately identify the messages aging faster, becomes more crucial in reducing the average instantaneous AoI. Furthermore, the difference between the average instantaneous AoI with learning and with full information also increases as $p_{i,t}$ increases. This is because it becomes increasingly difficult to learn the device types as $p_{i,t}$ increases.\\
\indent For the distributed RB allocation scheme, the communication range $r_c$ determines the information $\boldsymbol{A}_i$ that the devices have. For the given dimensions of the deployment area, $r_c \geq 15$ m is sufficiently large such that $|\boldsymbol{A}_i| = |\boldsymbol{A}|$ for any device $i$. SCA algorithm is compared against two algorithms, which are the random RB selection and the pre-determined RB selection. The pre-determined RB selection scheme \cite{kolkata} is known as the dictator's solution in the KPR game, and the RB usage for a device $i$ is pre-determined based on the rank of $F_i$. For instance, if $F_i$ is $\kappa$-th highest in $\boldsymbol{A}_i$ for $\kappa \leq R$, device $i$ uses a specific RB as previously agreed among IoT devices. The pre-determined RB allocation scheme requires full information $|\boldsymbol{A}_i| = |\boldsymbol{A}|$ for all $i$ and always achieves a service rate $s_r$ of $1$. While the pre-determined RB selection is used for optimal performance comparison, the random RB selection is used for baseline comparison. To analyze different distributed RB allocation schemes, the activation probability $v_a$, the SNR outage probability $p_{i,t}$, and the communication range $r_c$ are varied. In addition to the average instantaneous AoI, the service rate $s_r$ is evaluated for different distributed RB allocation schemes.

\begin{figure}[t]       
\centering
\includegraphics[width = 9cm]{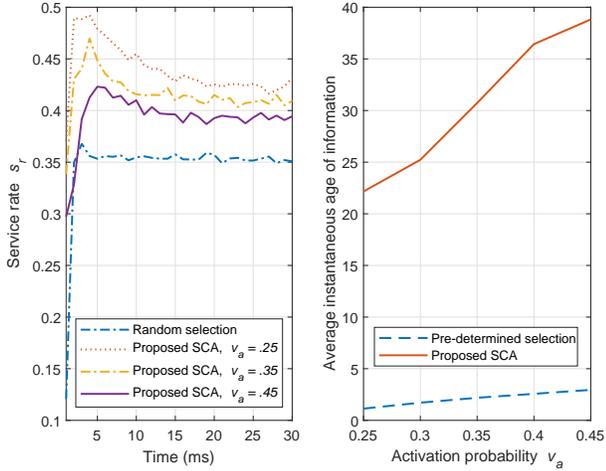}\vspace{-.2cm}
\caption{Average instantaneous AoI and service rate using distributed RB allocation schemes while varying $v_a$.}\vspace{-.5cm}
\label{dpa}
\end{figure}

Fig. \ref{dpa} shows the average instantaneous AoI and the service rate of the devices using distributed RB allocation schemes for different values of the activation probability $v_a$ with $p_{i,t} = 0.01$, $\epsilon  = 1$, $r_c = 10$ m, and $N = 200$. It is important to note that the expected number of newly active devices at a given time slot is $Nv_a$, and, thus, the number of active devices $N_t$ outnumbers the number of RBs $R$ for high values of $v_a$, simulating a massive IoT. Moreover, with $r_c = 10$ m, SCA only has partial information such that $|\boldsymbol{A}_i| < |\boldsymbol{A}|$ for all $i$. The service rate converges to $0.42$ for SCA with $v_a = 0.25$, $0.41$ for SCA with $v_a = 0.35$, $0.39$ for SCA with $v_a = 0.45$, and $0.35$ for random RB allocation with $v_a = 0.35$. As $v_a$ increases from $0.25$ to $0.45$, the average instantaneous AoI increases from $22.16$ to $38.82$ using SCA. As $N_t$ increases with increasing $v_a$, the transmission failure is more likely to occur due to the duplicate RB selection, because $R$ is fixed. Therefore, as $v_a$ increases, $s_r$ decreases, and the average instantaneous AoI increases. The pre-determined RB allocation scheme achieves a much lower average instantaneous AoI compared to SCA, because the pre-determined RB allocation scheme requires and uses the full information. Furthermore, the average instantaneous AoI with random RB allocation is multiple orders of magnitude higher than the average instantaneous AoI with SCA. Therefore, in a massive IoT with partial information, the proposed SCA algorithm is the most suitable algorithm to achieve low average instantaneous AoI as it balances between having full information and performing arbitrary allocations.

\begin{figure}[t]       
\centering
\includegraphics[width = 9cm]{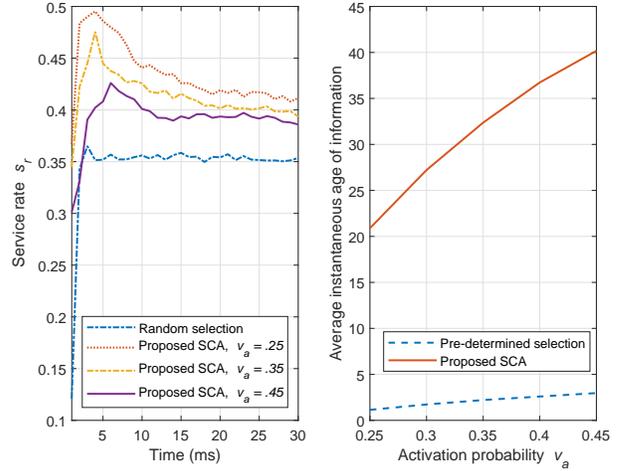}\vspace{-.2cm}
\caption{Average instantaneous AoI and service rate using distributed RB allocation schemes with different transmit powers while varying $v_a$.}\vspace{-.5cm}
\label{dpaxp}
\end{figure}

Fig. \ref{dpaxp} shows the average instantaneous AoI and the service rate of the devices with different transmit powers using distributed RB allocation schemes for different values of the activation probability $v_a$ with $p_{i,t} = 0.01$, $\epsilon = 1$, $r_c = 10$ m, and $N = 200$. The devices have different transmit powers such that their SNR values are uniformly distributed random variables from $17.0$ dB to $21.8$ dB, and the only difference between Fig. \ref{dpa} and Fig. \ref{dpaxp} is assumption on the transmit powers of devices. The difference in the service rates between Fig. \ref{dpa} and Fig. \ref{dpaxp} is insignificant. However, there is a notable increase in the average instantaneous AoI in Fig. \ref{dpaxp} compared to Fig. \ref{dpa}. Some devices have higher outage probability and other devices have lower outage probability, because the devices have different transmit powers. With an exponential aging function, an increase in the average instantaneous AoI with higher outage probability is more significant than a decrease in the average instantaneous AoI with lower outage probability. Therefore, similar to Fig. \ref{cpaxp}, there is a slight increase in the average instantaneous AoI when the devices have different transmit powers. Even when the devices have different transmit powers, the proposed SCA algorithm is still the most suitable algorithm to achieve low average instantaneous AoI in a massive IoT with partial information.

\begin{figure}[t]       
\centering
\includegraphics[width = 9cm]{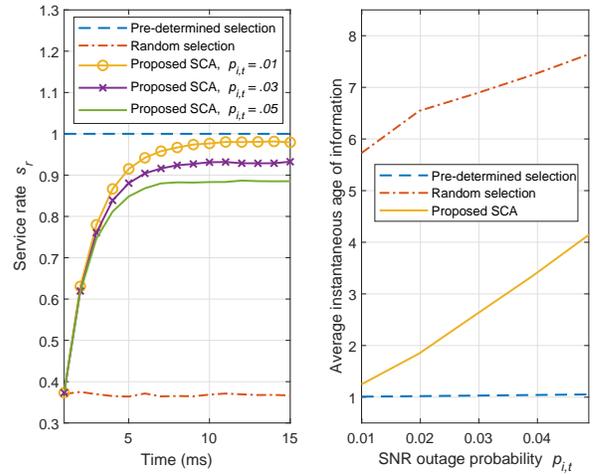}\vspace{-.2cm}
\caption{Average instantaneous AoI and service rate using distributed RB allocation schemes while varying $p_{i,t}$.}\vspace{-.5cm}
\label{deps}
\end{figure}

Fig. \ref{deps} shows the average instantaneous AoI and the service rate of the devices using distributed RB allocation schemes for different values of the SNR outage probability $p_{i,t}$ with $r_c = 10$ m, $v_a = 1$, $N = R = 50$, and varying $\epsilon$ from $1$ to $5$. It is important to note that the number of active devices $N_t$ is equal to $R$ with $N = R$ and $v_a = 1$, and this is the condition considered in Theorem \ref{thm2}. The service rate converges to $0.98$ for SCA with $p_{i,t} = 0.01$, $0.92$ for SCA with $p_{i,t} = 0.03$, $0.88$ for SCA with $p_{i,t} = 0.05$, and $0.37$ for random RB allocation. As $p_{i,t}$ increases from $0.01$ to $0.05$, the average instantaneous AoI increases from $1.25$ to $4.14$ using SCA, while the average instantaneous AoI increases from $5.73$ to $7.64$ using random RB allocation. With high $p_{i,t}$, the proposed SCA algorithm achieves about $45.8\%$ lower average instantaneous AoI when compared to random RB allocation. Since $T_t = N$ with $v_a = 1$ and $N = R$, the theoretical value of $s_r$ \eqref{srr} for random RB allocation case matches the simulated value of $s_r$ in Fig. \ref{deps}. As $p_{i,t}$ increases, the converged value of $s_r$ for SCA decreases, because the proposed SCA assumes that the transmission failures are caused by duplicate RB selection. Therefore, using SCA, a device $i$ stochastically avoids to use an RB even when there was no duplicate RB selection and the transmission failure is caused by SNR outage. Furthermore, the difference between the average instantaneous AoI using SCA and random RB allocation decreases as $p_{i,t}$ increases, because increasing $p_{i,t}$ has a more negative impact on SCA than on random RB allocation. However, it is important to note that SCA with low $p_{i,t}$ converges quickly to the service rate of $1$ with low $p_{i,t}$ as discussed in Theorem \ref{thm2}.

\begin{figure}[t]       
\centering
\includegraphics[width = 9cm]{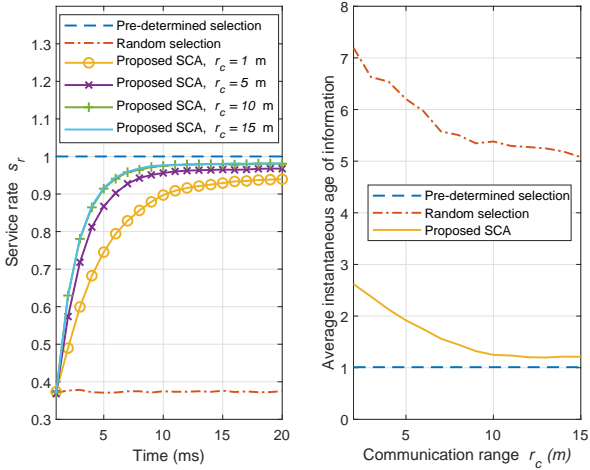}\vspace{-.2cm}
\caption{Average instantaneous AoI and service rate using distributed Rb allocation schemes while varying $r_c$.}\vspace{-.5cm}
\label{drc}
\end{figure}

Fig. \ref{drc} shows the average instantaneous AoI and the service rate of the devices using distributed RB allocation schemes for different values of the communication range $r_c$ with $p_{i,t} = 0.01$, $\epsilon  = 1$, $v_a = 1$, and $N = R = 50$. The communication range $r_c$ determines the amount of information $\boldsymbol{A}_i$ that the devices have, and $r_c = 15$ m implies that the devices have full information $|\boldsymbol{A}_i| = |\boldsymbol{A}|$ for all $i$. The service rate converges to $0.981$ for SCA with $r_c = 15$ m, $0.978$ for SCA with $r_c = 10$ m, $0.967$ for SCA with $r_c = 5$ m, $0.939$ for SCA with $r_c = 1$ m, and $0.374$ for random RB allocation. As $r_c$ increases from $2$ m to $15$ m, the average instantaneous AoI decreases from $2.62$ to $1.21$ using SCA, while the average instantaneous AoI decreases from $7.19$ to $5.07$ using random RB allocation. With low $r_c$, the proposed SCA algorithm achieves about $63.6\%$ lower average instantaneous AoI when compared to random RB allocation. Similar to the Fig. \ref{deps}. the theoretical and simulated values of $s_r$ for random RB allocation are matched. As $r_c$ increases, the convergence value of $s_r$ for SCA increases, because the devices have more information $\boldsymbol{A}_i$ with higher $r_c$. As $r_c$ increases sufficiently such that $|\boldsymbol{A}_i| = |\boldsymbol{A}|$ for all $i$, the service rate converges to $1$ as discussed in Theorem \ref{thm2}. However, SCA with only partial information can still achieve $s_r$ close to $1$. Moreover, as $r_c$ increases, the average instantaneous AoI using SCA converges to the average instantaneous AoI using pre-determined RB allocation scheme.\\
\indent Next, the proposed centralized and distributed RB allocation schemes are compared in a massive IoT with $N > R$ and in an ideal IoT described in Theorem \ref{thm2}. To analyze different RB allocation schemes, the activation probability $v_a$ and the SNR outage probability $p_{i,t}$ are varied,

\begin{figure}[t]       
\centering
\includegraphics[width = 9cm]{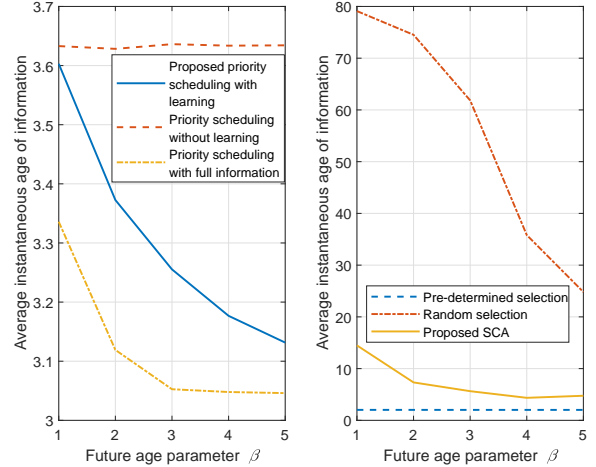}\vspace{-.2cm}
\caption{Average instantaneous AoI using centralized and distributed RB allocation schemes while varying $\beta$ in a massive IoT.}\vspace{-.5cm}
\label{beta}
\end{figure}

Fig. \ref{beta} shows the average instantaneous AoI of the devices using centralized and distributed RB allocation schemes for different values of $\beta$ with $p_{i,t} = 0.01$, $\epsilon = 1$, $r_c = 10$ m, $N = 100$, and $v_a = 1$. This simulates a massive IoT as the number of devices $N$ greatly outnumbers the number of RBs $R$. For the centralized RB allocation scheme, as $\beta$ increases, the average instantaneous AoI with learning decreases from $3.6$ to $3.13$, and the average instantaneous AoI with full information decreases from $3.34$ to $3.05$. However, the average instantaneous AoI without learning does not change significantly. With many type $2$ devices frequently transmitting exponentially aging messages, high values of $\beta$ can effectively reduce the average instantaneous AoI of type $2$ devices as BS learns the device types. However, without learning the device types, high values of $\beta$ are ineffective in reducing the average instantaneous AoI. For the distributed RB allocation scheme, as $\beta$ increases, the average instantaneous AoI with random selection decreases from $79.5$ to $25.1$, and the average instantaneous AoI with proposed SCA decreases from $13.4$ to $5.1$. This is because high values of $\beta$ enable the exponentially aging messages to be transmitted before their AoI increases greatly due to duplicate RB selection. However, the average instantaneous AoI with pre-determined selection does not change significantly, because most duplicate RB selection can be avoided with given $r_c$ and pre-determined selection. $\beta$ affects centralized and distributed RB selection schemes differently depending on other parameters of the IoT.

\begin{figure}[t]       
\centering
\includegraphics[width = 9cm]{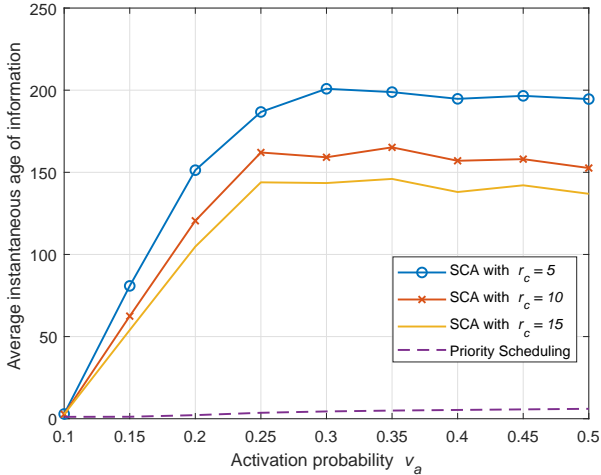}\vspace{-.2cm}
\caption{Average instantaneous AoI using centralized and distributed RB allocation schemes while varying $v_a$ in a massive IoT.}\vspace{-.5cm}
\label{cdpa}
\end{figure}

Fig. \ref{cdpa} shows the average instantaneous AoI of the devices using centralized and distributed RB allocation schemes for different values of the activation probability $v_a$ with $p_{i,t} = 0.01$, $\epsilon = 1$, and $N = 200$. This simulates a massive IoT as the number of devices $N$ greatly outnumbers the number of RBs $R$. As $v_a$ increases, the average instantaneous AoI converges to $200$ for SCA with $r_c = 5$ m, $157$ for SCA with $r_c = 10$ m, and $145$ for SCA with $r_c = 15$ m. On the other hand, for priority scheduling, the average instantaneous AoI increases from $1.14$ to $6.03$ as $v_a$ increases from $0.1$ to $0.5$. With high $v_a$, the proposed SCA algorithm with $r_c = 15$ m achieves about $24$-fold higher average instantaneous AoI when compared to the proposed priority scheduling with learning. For almost any values of $v_a$, centralized RB allocation with priority scheduling performs much better than distributed RB allocation with SCA in terms of the average instantaneous AoI. However, centralized RB allocation scheme requires the BS to dictate the RB allocation for all devices, and, thus, centralized RB allocation scheme may not be viable for some of the IoT. Moreover, similar to Fig. \ref{cpa}, the average instantaneous AoI flattens after a certain value of $v_a$, because all RBs are fully saturated. There is a performance gap between SCA with different values of $r_c$, because $r_c$ is directly related to the amount of information that the devices have. With higher $r_c$ and more information for the devices, SCA is more effective in reducing the average instantaneous AoI. 

\begin{figure}[t]       
\centering
\includegraphics[width = 9cm]{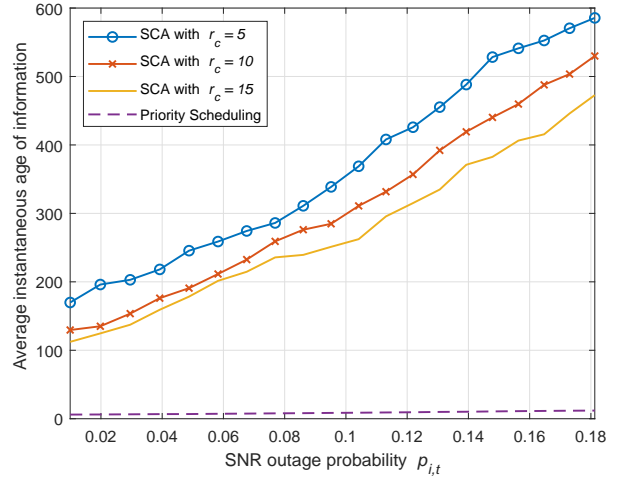}\vspace{-.2cm}
\caption{Average instantaneous AoI using centralized and distributed RB allocation schemes while varying $p_{i,t}$ in a massive IoT.}\vspace{-.5cm}
\label{cdeps}
\end{figure}

Fig. \ref{cdeps} shows the average instantaneous AoI of the devices using centralized and distributed RB allocation schemes for different values of the SNR outage probability $p_{i,t}$ with $v_a = 0.5$, $N = 200$, and varying $\epsilon$ from $1$ to $20$. Similar to Fig. \ref{cdpa}, this simulates a massive IoT. When $p_{i,t} = 0.18$. the average instantaneous AoI is $585.59$ using SCA with $r_c = 5$ m, $529.94$ using SCA with $r_c = 10$ m, $472.81$ using SCA with $r_c = 15$ m, and $11.90$ using priority scheduling. With high $p_{i,t}$, the proposed SCA algorithm with $r_c = 15$ m achieves about $40$-fold higher average instantaneous AoI when compared to the proposed priority scheduling with learning. Similar to Fig. \ref{cdpa}, centralized RB allocation with priority scheduling performs much better than distributed RB allocation with SCA in terms of the average instantaneous AoI for all values of $p_{i,t}$. It is interesting to note that the difference in the average instantaneous AoI between SCA algorithms increases as $p_{i,t}$ increases. This implies that SCA with less information is more severely affected by increasing $p_{i,t}$ than SCA with more information. 

\begin{figure}[t]       
\centering
\includegraphics[width = 9cm]{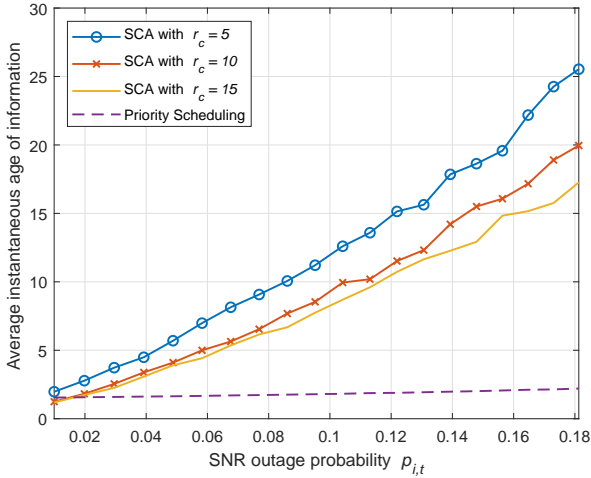}\vspace{-.2cm}
\caption{Average instantaneous AoI using centralized and distributed RB allocation schemes while varying $p_{i,t}$ in an ideal IoT.}\vspace{-.5cm}
\label{cdid}
\end{figure}

Fig. \ref{cdid} shows the average instantaneous AoI of the devices using centralized and distributed RB allocation schemes for different values of the SNR outage probability $p_{i,t}$ with $v_a = 1$, $N = R = 50$, and varying $\epsilon$ from $1$ to $20$. This simulates an ideal IoT for SCA as some of the conditions for NE convergence in Theorem \ref{thm2} are satisfied. When $p_{i,t} = 0.18$. the average instantaneous AoI is $25.53$ using SCA with $r_c = 5$ m, $19.96$ using SCA with $r_c = 10$ m, $17.27$ using SCA with $r_c = 15$ m, and $2.21$ using priority scheduling. With high $p_{i,t}$, the proposed SCA algorithm with $r_c = 15$ m achieves about $8$-fold higher average instantaneous AoI when compared to the proposed priority scheduling with learning. It is interesting to note that even in an ideal IoT for SCA, priority scheduling with learning performs better than SCA in terms of the average instantaneous AoI for most values of $p_{i,t}$. Moreover, similar to Fig. \ref{cdeps}, the difference in the average instantaneous AoI between SCA algorithms increases as $p_{i,t}$ increases.\\
\indent From our simulations, we observe that both priority scheduling and SCA are susceptible to high SNR outage probability $p_{i,t}$ as the average instantaneous AoI increases without flattening as $p_{i,t}$ increases. This is because the SNR outage probability is directly related to the transmission failures. However, the average instantaneous AoI increases slowly after a certain value of the activation probability $v_a$, because the RBs are fully saturated. Since increasing $v_a$ with fixed $N$ is equivalent to increasing $N$ with fixed $v_a$, the average instantaneous AoI also flattens for the case in which only the number of devices $N$ increases. Although centralized RB allocation scheme outperforms distributed RB allocation scheme in most cases, SCA can still achieve a high service rate $s_r$ and low average instantaneous AoI only with partial information. Furthermore, communication range $r_c$ and information availability are critical to the performance of SCA. 
\section{Conclusion}
In this paper, we have proposed centralized and distributed approaches for allocating the limited communication resources based on the aging function and the current AoI of IoT devices. In the presence of both linear and exponential aging functions, we have shown that comparing the future AoI achieves a lower average instantaneous AoI at the BS than comparing the current AoI. For the centralized approach, we have introduced a priority scheduling scheme with learning, which enables the BS to allocate the limited RBs to the heterogeneous devices based on their future AoI. For the distributed approach, we have formulated the problem of autonomously allocating the limited RBs to the devices using game theory, and we have designed payoff functions to encourage the devices with high AoI to transmit, while discouraging the devices with low AoI to not transmit. Furthermore, we have proposed a novel SCA algorithm such that the heterogeneous devices can allocate the RBs in a self-organizing manner to avoid the duplicate RB selection and to minimize the AoI. We have proved the conditions that a vector of actions in the IoT game must satisfy to achieve an NE. Furthermore, we have proved that the actions of devices using our proposed SCA algorithm converge to an NE, if the devices have sufficient information under certain network parameters. Simulation results have shown that the average instantaneous AoI is an increasing function of the activation probability and the SNR outage probability. Moreover, the simulation results have shown that the service rate is an increasing function of the communication range and a decreasing function of the activation probability and the SNR outage probability. We have compared our centralized and distributed RB allocation schemes, and we have shown that our centralized RB allocation scheme outperforms our distributed RB allocation scheme in most cases. However, our proposed SCA algorithm has shown to be effective in reducing the AoI and increasing the service rate only with partial information. With high SNR outage probability, the proposed priority scheduling scheme with learning has shown to achieve about $31.7\%$ lower average instantaneous AoI when compared to simple priority scheduling scheme. Furthermore, with high SNR outage probability, the proposed SCA algorithm has shown to achieve about $45.8\%$ lower average instantaneous AoI when compared to random RB allocation.

\section*{Acknowledgment}
This research was supported by the U.S. Office of Naval Research (ONR) under Grant N00014-19-1-2621.

\appendix \vspace{-0.0cm}
\subsection{Proof of Proposition 1} \label{app1}
Without a loss of generality, let there be an IoT device $i$ with aging function $a_i(t)$ and an IoT device $h$ with aging function $b_h(t)$. Moreover, the proof only considers the case with $N = 2$ and $R = 1$. It is sufficient to only consider $N = 2$, because comparing the AoI of $N > 2$ devices is equivalent to doing pairwise AoI comparison $\sfrac{N(N-1)}{2}$ times. It is unnecessary to consider the cases of $R = 0$ and $R \geq 2$. If $N = 2$ and $R = 0$, none of the devices can be allocated an RB. If $N = 2$ and $R \geq 2$, all devices can be allocated the RBs. Therefore, the AoI comparison to determine RB allocation is unnecessary.\\
\indent At time slot $\tau$, the current AoI of devices $i$ and $h$ must be one of the following cases: $a_i(\tau) > b_h(\tau)$, $a_i(\tau) = b_h(\tau)$, or $a_i(\tau) < b_h(\tau)$. If $a_i(\tau) \leq b_h(\tau)$, then $a_i(\tau+\beta) \leq b_h(\tau+\beta)$ for any positive integer $\beta$. At time slot $\tau$, the current AoI comparison with $t = \tau$ and the future AoI comparison with $t = \tau + \beta$ are equivalent, because device $h$ is allocated with an RB for both. However, if $a_i(\tau) > b_h(\tau)$, the future AoI of devices $i$ and $h$ can be any of the following cases: $a_i(\tau+\beta) > b_h(\tau+\beta)$, $a_i(\tau+\beta) = b_h(\tau+\beta)$, or $a_i(\tau+\beta) < b_h(\tau+\beta)$. For the case of $a_i(\tau) > b_h(\tau)$ and $a_i(\tau+\beta) \geq b_h(\tau+\beta)$, the current AoI comparison with $t = \tau$ and the future AoI comparison with $t = \tau + \beta$ are equivalent, because device $i$ is allocated with an RB for both.\\
\indent Comparing the current AoI and the future AoI are different if $a_i(\tau) > b_h(\tau)$ and $a_i(\tau+\beta) < b_h(\tau+\beta)$, because comparing the current AoI allocates the RB to device $i$, while comparing the future AoI allocates the RB to device $h$. When one device is allocated an RB at time slot $\tau$ and the other device is allocated an RB at time slot $\tau + \beta$, the RB allocation based on current AoI achieves the average instantaneous AoI of $0.5(a_i(\tau) + 2^{\beta}b_h(\tau))$, and the RB allocation based on future AoI achieves the average instantaneous AoI of $0.5(a_i(\tau)+\beta+b_h(\tau))$. For any $\beta \in \mathbb{Z}_+$, comparing the average instantaneous AoI of two cases is:\vspace{-1mm}
\begin{align}
    \frac{a_i(\tau) + 2^{\beta}b_h(\tau)}{2} &> \frac{a_i(\tau)+\beta+b_h(\tau)}{2},\\
    (2^{\beta} - 1) b_h(\tau) &> \beta,\\
    b_h(\tau) &> \frac{\beta}{2^{\beta} - 1}.
\end{align}
Even when one device is allocated with an RB at time slot $\tau$ and the other device is allocated with an RB one time slot later at time slot $\tau + 1$, the current AoI comparison yields higher average instantaneous AoI than the future AoI comparison. Since $\frac{\beta}{2^{\beta} - 1} = 1$ with $\beta = 1$, $b_h(\tau)$ cannot be less than or equal to 1, because the condition of $a_i(\tau) > b_h(\tau)$ and $a_i(\tau+1) < b_h(\tau+1)$ cannot be satisfied. Therefore, at time slot $\tau$, comparing the future AoI with $t = \tau + \beta$ to determine the RB allocation achieves lower average instantaneous AoI than comparing the current AoI with $t = \tau$. \vspace{-1mm}
\subsection{Proof of Theorem 1} \label{app2}
\indent With $\boldsymbol{x}(t)$ satisfying the given conditions, then at most $R$ active devices with sufficiently high values of $F_i$ are transmitting successfully, while rest of the devices are not transmitting. Assuming that all other devices do not change their action, an active device $i$ that is transmitting successfully with sufficiently high $F_i$ cannot change its action $x_i(t)$ to get higher payoff than its current payoff of $\rho$. If device $i$ uses some other RB, then its transmission may fail due to the duplicate RB usages, getting the payoff of $-\gamma$, or its transmission may succeed, getting the same payoff of $\rho$. If device $i$ does not transmit, then the payoff is $-(\gamma + \eta)$ as $F_i$ is greater than $\boldsymbol{A}_i(R)$. Therefore, the active devices that are transmitting successfully with sufficiently high $F_i$ do not change their action.\\
\indent The devices that are not transmitting may be active or inactive. An inactive device that is not transmitting has a payoff of $(\rho + \eta)$, which is higher than the payoff of transmitting successfully. Therefore, the inactive devices do not transmit. With $\boldsymbol{x}(t)$ satisfying the given conditions, then active devices with $F_i < \boldsymbol{A}_i(R)$ are not transmitting. With $y_i(\boldsymbol{x}(t))$, the active devices with $F_i < \boldsymbol{A}_i(R)$ have the payoff of $(\rho + \eta)$, which is higher than the payoff of transmitting successfully. Therefore, the active devices with $F_i < \boldsymbol{A}_i(R)$ do not change their action from not transmitting.\\
\indent With the design of payoff function $y_i(\boldsymbol{x}(t))$ \eqref{payoff3}, inactive devices and active devices with $F_i < \boldsymbol{A}_i(R)$ have the highest payoff of $(\rho +\eta)$ by not transmitting. Moreover, with $|\boldsymbol{A}_i| = |\boldsymbol{A}|$, at most $R$ active devices with $F_i \geq \boldsymbol{A}_i(R)$ have higher payoff from transmitting successfully than from not transmitting. With $\boldsymbol{x}(t)$ such that each of the RBs is used by at most one device, the active devices with $F_i \geq \boldsymbol{A}_i(R)$ have the highest payoff of $\rho$ by transmitting successfully. Therefore, with $y_i(\boldsymbol{x}(t))$, any vector of actions $\boldsymbol(x)(t)$ such that at most $R$ active devices with $F_i \geq \boldsymbol{A}_i(R)$ transmit, rest of the devices do not transmit, and each of the RBs is used by at most one device is an NE.\vspace{-1mm}
\subsection{Proof of Theorem 2} \label{app3}
\indent At time slot $t = 1$, $\boldsymbol{x}(t)$ is initialized as a random RB selection. For $N = R$ and $v_a = 1$, $T_t = N$ and the service rate $s_r$ is $\left(\sfrac{(R - 1)}{R}\right)^{N-1}$. Therefore, from the initial RB allocation, the expected number of RBs that are used by one device is $Rs_r$. With SCA and $v_a = 1$, an RB that is used by one device at time slot $t = 1$ is used by the same device at time slot $t = 2$. Furthermore, with SCA, the RBs that are used by more than one device at the time slot $t = 1$ are expected to be used by one device at the time slot $t = 2$. Therefore, at time slot $t = 2$, $\mathbb{E}[|\mathcal{L}|] = R\left(\sfrac{(R-1)}{R}\right)^{N}$, which is the expected number of RBs that are used by none of the devices at time slot $t = 1$. $\mathbb{E}[|\mathcal{L}|]$ also is the expected number of devices that are competing to use the RBs in $\mathcal{L}$, because the devices use at most $1$ RB at each time slot. Since the devices choose the RBs in $\mathcal{L}$ randomly, the RB selection at time slot $t = 2$ is equivalent to the random RB selection with the number of devices and RBs equal to $R\left(\sfrac{(R-1)}{R}\right)^{N}$. Moreover, the same analysis done for $t = 1$ can be done with $t = 2$.\\
\indent Expanding to the general case, at time slot $t$, the expected number of RBs that are used by none of the devices and the expected number of devices competing to use the RBs in $\mathcal{L}$ is $R\left(\sfrac{(R-1)}{R}\right)^{N(t-1)}$. As $t$ increases to infinity, the expected number of RBs that are used by none of the devices and the expected number of devices competing to use the RBs in $\mathcal{L}$ decrease to $0$. With SCA, this implies that the number of RBs each used by one device increases to $R$ as $t$ increases to infinity. Furthermore, with $N = R$ and $|\boldsymbol{A}_i| = |\boldsymbol{A}|$ for all $i$, any active device $i$ satisfies $F_i \geq \boldsymbol{A}_i(R)$. The action space $\boldsymbol{x}(t)$ converged using SCA is such that all $R$ active devices transmit and each of the RBs is used by one device. Therefore, $\boldsymbol{x}(t)$ converged using SCA is an NE.\vspace{-1mm}
\subsection{Proof of Proposition 2} \label{app4}
When the $T_t$ transmitting devices use random RB selection, the service rate is equivalent to the probability of an RB being used by only one transmitting device. Therefore, the service rate $s_r$ is:
\begin{equation}
    s_r = \binom{T_t}{1} \frac{1}{R} \left(1 - \frac{1}{R}\right)^{T_t - 1} = \frac{T_t}{R}\left(\frac{R - 1}{R}\right)^{T_t - 1}.
\end{equation}
When the number of IoT devices $N$ increases to infinity in a massive IoT, the number of transmitting devices $T_t$ also increases to infinity with fixed $v_a$, and the service rate $s_r$ is:
\begin{align}
    &\lim_{T_t \rightarrow \infty} s_r =  \lim_{T_t \rightarrow \infty} \frac{T_t}{R} \left(1 + \frac{\sfrac{-T_t}{R}}{T_t}\right)^{T_t - 1},\\
    & \ \ \ = \lim_{T_t \rightarrow \infty} \frac{T_t}{R-1} \left(1 + \frac{\sfrac{-T_t}{R}}{T_t}\right)^{T_t} = \frac{T_t}{R-1}\exp\left(\frac{-T_t}{R}\right). 
\end{align}

\bibliographystyle{IEEEtran}
\bibliography{references}
\end{document}